\documentclass[onecolumn,draftcls]{IEEEtran}

\usepackage{latexsym}
\usepackage{bm}
\usepackage[dvips]{graphics}
\usepackage{graphicx}
\usepackage{psfrag}
\usepackage{amsfonts,amsmath,amssymb}
\usepackage{dsfont,color,subfigure}

\DeclareMathOperator*{\argmax}{arg\,max}

\usepackage{xspace}
\usepackage{bbm}

\def\In{\mathop{\rm In}\nolimits}%
\def\Out{\mathop{\rm Out}\nolimits}%
\newcommand{\Ac}{\mathcal{A}}
\newcommand{\Bc}{\mathcal{B}}

\newcommand{\Dc}{\mathcal{D}}
\newcommand{\Ec}{\mathcal{E}}

\newcommand{\Hc}{\mathcal{H}}
\newcommand{\Ic}{\mathcal{I}}
\newcommand{\Kc}{\mathcal{K}}
\newcommand{\Mc}{\mathcal{M}}
\newcommand{\Nc}{\mathcal{N}}
\newcommand{\Oc}{\mathcal{O}}

\newcommand{\Rc}{\mathcal{R}}
\newcommand{\Sc}{\mathcal{S}}
\newcommand{\Tc}{\mathcal{T}}
\newcommand{\Uc}{\mathcal{U}}
\newcommand{\Vc}{\mathcal{V}}
\newcommand{\Wc}{\mathcal{W}}
\newcommand{\Xc}{\mathcal{X}}
\newcommand{\Yc}{\mathcal{Y}}

\newcommand{\Uv}{{\bf U}}

\newcommand{\xv}{{\bf x}}
\newcommand{\yv}{{\bf y}}
\newcommand{\hv}{{\bf h}}

\newcommand{\rv}{{\bf r}}

\newcommand{\Mh}{{\hat{M}}}

\newcommand{\Uh}{{\hat{U}}}

\newcommand{\mh}{{\hat{m}}}

\newcommand{\Mt}{{\tilde{M}}}

\newcommand{\Ut}{{\tilde{U}}}

\newcommand{\Yt}{{\tilde{Y}}}

\newcommand{\yt}{{\tilde{y}}}

\def\a{\alpha}
\def\b{\beta}
\def\g{\gamma}
\def\d{\delta}
\def\e{\epsilon}

\def\l{\lambda}

\DeclareMathOperator\E{E}
\let\P\relax
\DeclareMathOperator\P{P}

\newcommand\ie{i.e.,\xspace}
\def\textiid{i.i.d.\@\xspace}
\newcommand\iid{\ifmmode\text{ i.i.d. } \else \textiid \fi}

\newcommand{\ind}{\mathbbmss{1}}

\newtheorem{definition}{Definition}

\newtheorem{remark}{Remark}

\newtheorem{theorem}{Theorem}
\newtheorem{lemma}{Lemma}

\newtheorem{corollary}{Corollary}

\begin{document}

\title{On Capacity Region of Wiretap Networks}

\author{Shirin Jalali and Tracey Ho}
\maketitle

\newcommand{\p}{\mathds{P}}
\newcommand{\Lc}{\mathc al{L}}
\newcommand{\Jc}{\mathcal{J}}
\newcommand{\mb}{\mathbf{m}}
\newcommand{\bb}{\mathbf{b}}
\newcommand{\Xb}{\mathbf{X}}
\newcommand{\Yb}{\mathbf{Y}}
\newcommand{\Ub}{\mathbf{U}}
\newcommand{\La}{\Lambda}
\newcommand{\su}{\underline{s}}
\newcommand{\xu}{\underline{x}}
\newcommand{\yu}{\underline{y}}
\newcommand{\Xu}{\underline{X}}
\newcommand{\Yu}{\underline{Y}}
\newcommand{\Uu}{\underline{U}}
\newcommand{\Uhv}{\hat{\Uv}}
\newcommand{\Utv}{\mathbf{\Ut}}
\newcommand{\Ytv}{\mathbf{\Yt}}
\newcommand{\ytv}{\mathbf{\yt}}

\begin{abstract}

In this paper we study the problem of secure communication over networks in which an unknown subset of nodes or links can be wiretapped. We first consider  a general multi-terminal discrete memoryless network (DMN) described by its channel transition function from network inputs to network outputs, observed by network nodes and the eavesdropper. We prove that in such general networks with multiple sources and sinks,  the capacity regions subject to strong and weak secrecy requirements are equal. We then focus on the special case of  noiseless wiretap networks, i.e.,  wired  networks  of noiseless point-to-point directed channels, where an unknown subset of links, selected from a known  collection of such subsets,  can be wiretapped.  We derive  inner and outer bounds on the  capacity regions of such networks   in terms of the entropic region, for both zero probability of  error  and asymptotically zero  probability of error.

\end{abstract}


\section{Introduction}\label{sec:intro}

Consider the problem of  secure communication over multi-terminal  discrete memoryless networks (DMN), where the communication channel is described by a general transition function from the network nodes inputs to their outputs.  An adversary eavesdrops on an unknown subset of nodes selected from a known collection. The goal is to maximize the transmitted data rates while leaking as little information as possible to the eavesdropper. 

Different notions of secrecy have been developed in the literature to  quantify the information leakage to the adversary. In general, most of such notions can be categorized as  either computational   or information theoretic.  In this paper, we focus on the latter, where the information leakage is usually measured in terms of  the mutual information between the adversary's observations and the source messages. Depending on the convergence rate of the defined mutual information to zero, the system is said to achieve perfect secrecy, strong secrecy or weak secrecy. If the mutual information between the source messages and the adversary's observations is exactly equal to zero, then the system is perfectly secure. In that case, the source messages and the information available to the adversary are completely independent. In the case of strong secrecy, the described mutual information is required to get arbitrarily small as the blocklength grows, which implies almost independence. Finally, in the weak secrecy model, the mutual information is normalized by the blocklength, and hence the ``rate'' of  information leakage is required to be arbitrarily small.

The problem of information theoretically secure communication was originally introduced  by Shannon in \cite{Shannon:49}. He considered  a point to point public channel, where a transmitter  wants to send message $M$ uniformly distributed over $\{1,\ldots,2^{nR_s}\}$ to a receiver. A private random key   uniformly distributed over $\{1,\ldots,2^{nR_k}\}$ is available to both of them.  The transmitter encodes its message $M$ using key $K$ into codeword $W=W(M,K)$.  The adversary observes the encoded codeword $W$ sent over the public channel, but does not have access to $K$. Shannon proved that to ensure perfect secrecy, \ie $I(W;M)=0$, the rate of the key, $R_k$, should  at least be as large as the source data rate $R_s$. This rate is also sufficient; in fact, letting $R_k=R_s$, a simple linear code, $W=M \oplus K$, achieves the capacity in this model.

Secure communication over a wiretap channel was  defined by Wyner in \cite{Wyner:75}. There the transmitter attempts to send a message to the receiver over a broadcast channel described by $p(y,z|x)$. The intended  receiver observes output $y$ of the channel, while an adversary observes output $z$. In \cite{CsiszarK:78}, it was shown that the secrecy capacity, $C_s$, of the wiretap channel subject to the weak secrecy constraint is equal to $C_s=\max_{p(u,x)}(I(U;Y)-I(U;Z))$. Subsequently it was shown in  \cite{MaurerW:00} that in fact $C_s$ is also the secrecy capacity of the wiretap channel subject to the strong secrecy requirement, \ie $I(M;Z^n)\leq \e$ instead of $I(M;Z^n)\leq n\e$.

The problem of secure communication over  \emph{wiretap networks}   was  proposed and  studied  in \cite{CaiY:02}. In a wiretap network,  data communication takes place over a noiseless wired network of point-to-point directed links. An  adversary observes the messages sent over  a  subset of links. In \cite{CaiY:02} and its subsequent work \cite{CaiY:11}, the authors  considered a multicast problem where all messages and keys are generated at a single source, and defined secure network coding capacity subject to  zero-error communication and  perfect secrecy. When all the links  have unit capacity and the adversary can choose to access any subset of links of size at most $r$, they proved  that the secrecy capacity  is equal to  $c_s=n-r$, where $n$ denotes the max-flow  capacity from the source to all sinks. Moreover, they proved that linear codes are sufficient to achieve $c_s$. (Refer to \cite{CaiC:11-ieee} for  a comprehensive review of the literature on   linear network coding subject to security constraints.)

The problem of general multi-source multi-destination secure network coding  was studied in \cite{ChanG:08}, which considered asymptotically zero error communication subject to weak secrecy constraint  and presented inner and outer bounds on the set of achievable rates. The outer bound  and the inner bound were presented in terms of ``almost entropic pseudo-variables''  and random variables, respectively, each satisfying certain  conditions.  In \cite{ChanG:12-arxiv} the authors present an outer bound in terms of entropic region on the network coding capacity region subject to weak secrecy requirement.

In this paper we extend the above results in several directions. We first consider  a general DMN, described by a general channel transition function from network inputs to outputs observed by network nodes and the eavesdropper. We prove that in such general networks, with multiple sources and multiple sinks,  the capacity regions subject to weak and strong  secrecy are equal.  DMNs are often used to model general wireless networks. Hence, this result establishes  the equivalence of  the capacity regions  subject to strong and weak secrecy requirements in communication of independent sources over general noisy memoryless channels. This result  is a generalization of \cite{MaurerW:00}, where the same result was proved for a single wiretap channel. We extend the approach of \cite{MaurerW:00} to general networks with arbitrary communication demands among multiple terminals. The main ingredient of the proof in \cite{MaurerW:00} is ``extractor functions'', a well-studied tool in computer science. We employ  the same machinery here. The main difference between our work and \cite{MaurerW:00}  is that here we consider a general secure network coding problem with multiple sources and multiple sinks. As a corollary, our result also proves  that the network coding capacity regions of noiseless  wiretap networks  with asymptotically zero probability of error subject to strong secrecy   and  subject to weak secrecy  are equal.

We then focus on the special case of  wiretap networks, defined as an error-free network of  directed point-to-point links with an adversary eavesdropping  on an unknown subset of links chosen from a given collection of subsets. A typical example of such  collections is the set of all up to $r$-link subsets of edges. The selected edges are not known to the communicating users, and the goal is to maximize the communication rates while not leaking  information to the eavesdropper. For such networks we derive inner and outer bounds on the zero-error capacity region subject to weak secrecy constraint in terms of ``entropic region''. As will be explained later, in this case we consider variable-length coding across the channels. We  also derive inner and outer bounds on the capacity region subject to weak secrecy constraint and asymptotically zero probability of decoding error in terms of entropic region. Our outer bound is tighter than the bound derived in  \cite{ChanG:12-arxiv}. Our first result, \ie equivalence of network coding capacity regions subject to weak and strong secrecy requirements, states that the  inner and outer bounds   also hold in the case of strong secrecy.

In \cite{DikaliotisY:12,Ted_thesis}, the authors consider networks of wiretap channels. In the case where the eavesdropper has access only to the output of one channel chosen from a known collection, they  show that the weak secrecy capacity region is equal to the weak secrecy capacity region of another  network derived by replacing each noisy wiretap channel by an equivalent noiseless model.   Our result proves that both for the wired network and wireless network the capacity regions do not shrink  by requiring strong secrecy. Therefore, it shows the equivalence established in  \cite{DikaliotisY:12,Ted_thesis} also holds under strong secrecy condition.

The organization of this paper is as follows. Section \ref{sec:notation} defines the notations used throughout the paper, and also summarizes the definitions and some relevant properties  of typical sequences and the entropic region. In Section \ref{sec:DMN}, we extend the result proved in \cite{MaurerW:00} for wiretap channels to general DMNs and  show that  the capacity regions of  a general DMN  subject to weak and strong secrecy requirements are the same.   Section \ref{sec:wiretap} reviews the noiseless  wiretap network model  and defines the secure network coding capacity region of such networks  under different notions of secrecy and also different  probability of error requirements. Section \ref{sec:c-region-0} and \ref{sec:c-region-e} present inner and outer bounds on the network coding capacity region under weak secrecy constraint, for zero error and asymptotically  zero error, respectively. Finally, Section \ref{sec:conclusion} concludes the paper.

\section{Notation and definitions}\label{sec:notation}

\subsection{Notation}

For integers $i\leq j$, let $[i:j]\triangleq\{i,i+1,\ldots,j\}$. Random variables are denoted by upper case letters such as $X$.  For  random variable $X$, let script letter $\Xc$ denote its alphabet set.  Given vector $(X_1,\ldots,X_n)$ and  set $\a\subseteq \{1,\ldots,n\}$, let $X_{\a}\triangleq(X_i:i\in\a)$. For a set $\Xc$, let $2^{\Xc}$ denote the set of all subsets of $\Xc$. For  set $\Ac$,    $\overline{\Ac}$ denotes its closure, and  $\overline{\rm con}(\Ac)$ denotes its convex closure, which is defined as the smallest closed and convex set containing $\Ac$.

For $\xv,\yv\in\mathds{R}^n$,  $\xv\leq \yv$, if and only if, $x_i\leq y_i$, for all $i=1,\ldots,n$. For $\Ac \subset\mathds{R}^{n}$, let
\[
\Lambda(\Ac)\triangleq \{\xv\in\mathds{R}^{n}:\; \xv \leq \yv, {\rm for}\;{\rm some }\; \yv\in\Ac\}.
\]
Also, define
\[
D(\Ac)\triangleq \{\a\xv: \a\in[0,1], \xv\in\Ac \}.
\]
For $\Ic\subset[1:n]$ and $\Ac\subset \mathds{R}^n$, define the projection operation ${\rm Proj}(\cdot)$ over the coordinated defined by $\Ic$ as
\[
{\rm Proj}_{\Ic}(\Ac)\triangleq \{\xv_{\Ic}:\; \xv \in\Ac\},
\]
where $\xv_{\Ic}=(x_i:\;i\in\Ic)$.

The total variation distance between probability distributions  $p_1$ and $p_2$ defined over set $\Xc$ is defined  as $d_{\rm TV}(p_1,p_2)\triangleq0.5\sum_{x\in\Xc}|p_1(x)-p_2(x)|$. The  ``min-entropy'' of random variable $X$ with pdf $p(x): \Xc\to\mathds{R}^{+}$, $H_{\infty}(X)$, is defined as
\[
H_{\infty}(X)\triangleq\min_{x\in\Xc} \log{1\over p(x)}.
\]

\subsection{Typical sequences}

For a sequence $x^n\in\Xc^n$, the empirical distribution of $x^n$ is defined as
\[
\pi(x|x^n)\triangleq {|\{i: x_i=x\}|\over n},
\]
for all $x\in\Xc$. For $(x^n,y^n)\in\Xc^n\times\Yc^n$, the joint empirical distribution of $(x^n,y^n)$ is defined as
\[
\pi(x,y|x^n,y^n)\triangleq {|\{i: (x_i,y_i)=(x,y)\}|\over n},
\]
for all $(x,y)\in\Xc\times\Yc$.

We adopt the definition of strongly typical sequences introduced in \cite{OrlitskyR:01}. Given random variable $X$ distributed as $p(x)$, $x\in\Xc$, and $\e\in(0,1)$, the set of $\e$-typical sequences of length $n$ is defined as
\[
\Tc_{\e}^{(n)}(X) \triangleq \{x^n:|\pi(x|x^n)-p(x)|\leq \e\cdot p(x), \forall \;x\in\Xc\}.
\]
Similarly, for random variables $(X,Y)$ jointly distributed as $p(x,y)$,
\[
\Tc_{\e}^{(n)}(X,Y) \triangleq \{(x^n,y^n):|\pi(x,y|x^n,y^n)-p(x,y)|\leq \e\cdot p(x,y), \forall \;(x,y)\in\Xc\times\Yc\}.
\]
Given discrete random variable $X$,
 \begin{align}
(1-\e) 2^{n(1-\e)H(X)}< | \Tc_{\e}^{(n)}(X)| < 2^{n(1+\e)H(X)},\label{eq:size-typical-set}
 \end{align}
 for $n$  large enough \cite{ElGamalK_book}.

Consider discrete random variable $X$, and assume that $Y=g(X)$, where $g: \Xc\to\Yc$ is a deterministic function.
\begin{lemma}\label{lemma:typical-g}
If $x^n\in\Tc_{\e}^{(n)}(X)$, and $y^n=g(x^n)$, \ie $y_i=g(x_i)$, then $y^n\in\Tc_{\e}^{(n)}(Y)$.
\end{lemma}

\begin{proof}
For $y\in\Yc$, define $\Xc_y\triangleq\{x\in\Xc: g(x)=y\}$. Since $g$ is a deterministic function, $(\Xc_y: y\in\Yc)$ forms a partition of $\Xc$. For every $y\in\Yc$, we have
\begin{align}
n^{-1}|\pi(y|x^n)-p(y)|&=|\pi(y|y^n)-\sum_{x\in\Xc_y}p(x)|\nonumber\\
&=\Big|\sum_{x\in\Xc_y}\pi(x|x^n)-\sum_{x\in\Xc_y}p(x)\Big|\nonumber\\
&\leq \sum_{x\in\Xc_y}|\pi(x|x^n)-p(x)|\nonumber\\
&\leq \sum_{x\in\Xc_y}\e\cdot p(x)\nonumber\\
&=\e\cdot p(y),
\end{align}
which, by definition, shows that $y^n\in\Tc_{\e}^{(n)}(Y)$.
\end{proof}

Lemma \ref{lemma:typical-g} simplifies our analysis in the next sections. Note that $X$ in Lemma \ref{lemma:typical-g} can also be a vector, $(X_1,\ldots,X_m)$, where $X_1,\ldots,X_m$ are jointly distributed finite-alphabet   random variables. In that case, for $Y=g(X_1,\ldots,X_m)$, by Lemma \ref{lemma:typical-g}, if $(x_1^n,\ldots,x_m^n)\in\Tc_{\e}^{(n)}(X_1,\ldots,X_m)$ and $y_i=g(x_{1,i},\ldots,x_{m,i})$, $i=1,\ldots,n$, then $y^n\in\Tc_{\e}^{(n)}(Y)$ as well.

\subsection{Entropic region}

In this section, following  the notations used in \cite{yeung}, we briefly review the definitions of  the entropy function  of random variables and the entropic region.

Assume that $(X_1,\ldots,X_n)$ are $n$ jointly distributed finite-alphabet random variables. Define $\hv\in\mathds{R}^{2^n}$ to denote the entropy function of $(X_1,\ldots,X_n)$ defined as follows.  For $\a\subseteq[1:n]$,
\[
h_{\a}\triangleq H(X_{\a})=H(X_i:\;i\in\a).
\]
For $\a=\emptyset$, $h_{\a}=0$.
Let $\Hc_{n}$ denote the set of all possible $2^n$ dimensional vectors whose components are labeled by the $2^n$ subsets of $[1:n]$.   A vector $\hv\in\Hc_n$ is called entropic, if and only if $\hv$ is equal to the entropy function of  $n$ jointly distrusted finite-alphabet random variables $(X_1,\ldots,X_n)$. The  entropic region $\Gamma_n^{*}$ is defined as the set of all possible entropic vectors of $n$ random variables, \ie
\[
\Gamma_n^{*} \triangleq \{\hv\in\Hc_n\;:\; \hv\;{\rm is }\;{\rm entropic}\}.
\]
While $\Gamma_n^{*} $ is not in general a convex set, it is known that its closure $\overline{\Gamma_n^{*}} $ is a convex cone \cite{yeung}.
Characterizing the entropic region is a fundamental open problem in information theory.

Entropy function is  a submodular function, \ie given jointly distributed discrete random variables $X_1,\ldots,X_n$, for any $\a,\b\subseteq [1:n]$,
\[
H(X_{\a\cap\b})+H(X_{\a\cup\b})\leq H(X_{\a})+H(X_{\b}),
\]

 Let $\Gamma_n\subseteq\Hc_n$ denote the set of vectors in $\Hc_n$ that satisfy the basic inequalities of Shannon's information measures \cite{ZhangY:98}, \ie
\[
\Gamma_n \triangleq \{\hv\in\Hc_n\;:\; i(\a\; ; \b|\g)\geq 0, \forall \;\a,\b,\g\subseteq[1:n]\},
\]
where $i(\a;\b|\g)\triangleq I(X_{\a};X_{\b}|X_{\g})$.
It is known that for $n=2$, $\Gamma_2=\Gamma_n^{*} $ and for $n=3$, $\Gamma_3^{*}\neq\Gamma_3$, but $\overline{\Gamma_3^{*}}=\Gamma_3$. But for $n>3$, it is known that $\overline{\Gamma_n^{*}} \subset \Gamma_n$ \cite{yeung}.


\section{Discrete memoryless networks}\label{sec:DMN}

Consider a general discrete memoryless network (DMN). (See Fig.~\ref{fig:dmn}.) Let $\Vc$ denote the set of  $m$ nodes in the network.  The channel is described by its transition function $\{p(y^m|x^m)\}_{(x^m,y^m)\in\Xc^m\times\Yc^m}$. Node $i\in\Vc$ communicates with other nodes in the network through controlling input $X_i$ of the channel and receiving output $Y_i$. Let $\Sc\subset\Vc$ denote the set of source nodes in the network. Also let $\Ac$ denote the collection of potential sets of adversaries. Each $\a\in\Ac$ represents a subset of  $\Vc$ that might be accessed by the adversary. The adversary  chooses one and only one  subset $\a\in\Ac$ and observes the outputs of the nodes in $\a$. Clearly, $\Sc\cap \a=\emptyset$, for $\a\in\Ac$. The goal is to minimize the data leakage to the adversary, while maximizing the throughput.

The coding operations are performed as follows. Each node $s\in\Sc$ observes message $M_s\in[1:2^{nr_s}]$, and desires to communicate this message to a subset of nodes $\Dc_s\subset\Vc\backslash s$. The messages are described to the destination nodes through a  code of block length $n$. At time $i\in[1:n]$, node $s\in\Sc$, based on its received signals up to time $i-1$ and its own message $M_s$ generates its channel input $X_{s,i}$ as $X_{s,i}=f_{s,t}(M_s, Y_s^{i-1})$, where
\[
f_{s,i}: [1:2^{nr_s}]\times \Yc_s^{i-1} \to \Xc_s.
\]
At time $i\in[1:n]$, node $v\in\Vc\backslash\Sc$ generates its channel input $X_{v,i}$ as a function of its received signals as  $X_{v,i}=f_{v,i}(Y_{v}^{i-1})$, where
\[
f_{v,i}: \Yc_v^{i-1} \to \Xc_v.
\]

\begin{figure}[t]
\begin{center}
\psfrag{p}{\footnotesize $\;\;\;\;\;\;p(y_1,\ldots,y_m|x_1,\ldots,x_m)$}
\psfrag{x}{\footnotesize $\to X_j $}
\psfrag{y}{\footnotesize $\leftarrow Y_j$}
\includegraphics[width=6.2cm]{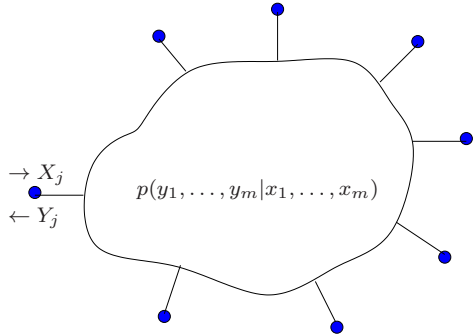} \caption{General discrete memoryless network (DMN)}\label{fig:dmn}
\end{center}
\end{figure}

Let   $\Mh_{s\to t}$ denote the reconstruction of message $M_s$ at node $t\in\Dc_s$. If node $t\in\Dc_s$ is also a source node, \ie $t\in\Sc$, then it reconstructs  $M_s$ as a function of its received signals $Y_t^n$ and  its own message $M_t$. In other words, $\Mh_{s\to t}=g_{s\to t}(Y_t^n,M_t)$, where
\[
g_{s\to t}:\Yc_t^n \times [1:2^{nr_t}]\to [1:2^{nr_s}].
\]
 If node $t\in\Dc_s$ is not a source node, \ie $t\notin\Sc$, then it reconstructs  $M_s$ only as a function of its received signals $Y_t^n$. In other words, $\Mh_{s\to t}=g_{s\to t}(Y_t^n)$, where
\[
g_{s\to t}:\Yc_t^n \to [1:2^{nr_s}].
\]

Coding rate vector $(r_s:\;s\in\Sc)$ is said to be achievable subject to ``weak'' secrecy, if for any $\e>0$ and  $n$ large enough, there exists a code of blocklength $n$  at message rates $(r^{(n)}_s: s\in\Sc)$, such that 
\begin{enumerate}
\item[i.] Source rates: For every $s\in\Sc$,
\[
 r^{(n)}_s\geq r_s-\e.
\]
\item[ii.] Reliable reconstructions: For every $s\in\Sc$ and $t\in\Dc_s$,
\[
\P(M_s\neq \Mh_{s\to t})\leq \e.
\]
\item[iii.] Secrecy: For every $\a\in\Ac$,
\[
{1\over n}I(Y_{\a}^n;M_{\Sc})\leq \e,
\]
where $M_{\Sc}=(M_s: s\in\Sc)$ and $Y_{\a}^n=(Y_v^n: v\in\a)$.
\end{enumerate}
Coding rate vector $(r_s:\;s\in\Sc)$ is said to be achievable subject to ``strong'' secrecy, if for any $\e>0$ and  $n$ large enough, there exists a code with blocklength $n$ and source rates $(r^{(n)}_s: s\in\Sc)$, such that the following constraints are satisfied:
\begin{enumerate}
\item[i.] Source rates: For $s\in\Sc$,
\[
r^{(n)}_s\geq r_s-\e.
\]
\item[ii.] Reliable reconstructions: For every $s\in\Sc$ and $t\in\Dc_s$,
\[
\P(M_s\neq \Mh_{s\to t})\leq \e.
\]
\item[iii.] Secrecy: For every $\a\in\Ac$,
\[
I(Y_{\a}^n;M_{\Sc})\leq \e.
\]
\end{enumerate}
Given an $m$-node DMN described by $p(y^m|x^m)$, source nodes $\Sc$ and collection of sets of potential adversaries $\Ac$, let  $\Rc^{(w)}(p(y^m|x^m),\Sc,\Ac)$ denote  its set of achievable rates (capacity region)  subject to   weak secrecy. Similarly,  let $\Rc^{(s)}(p(y^m|x^m),\Sc,\Ac)$ denote  its set of achievable rates (capacity region)  subject to strong secrecy.

\subsection{Strong versus weak secrecy}
In this section we prove that for any general DMN, the capacity region subject to weak secrecy is equal to the capacity region subject to strong secrecy. This result is proved in  \cite{MaurerW:00} for wiretap channels. Here, we extend this result to general DMNs with multiple sources and multiple sinks.

\begin{theorem}\label{thm:strong_eq_weak}
For any set of sources $\Sc\subset [1:m]$, and any collection of  sets of nodes $\Ac\subset 2^{[1:m]}$ that are observed by the adversary,
\[
\Rc^{(s)}(p(y^m|x^m),\Sc,\Ac)=\Rc^{(w)}(p(y^m|x^m),\Sc,\Ac).
\]
\end{theorem}

The proof of Theorem \ref{thm:strong_eq_weak} is presented in Appendix A. As done in \cite{MaurerW:00}, to prove this result, we employ   ``extractor functions''. Extractor functions are used in computer science to generate completely random bits with uniform distributions. An extractor function usually has two inputs,  some ``weakly'' random bits (source bits), and  some truly random bits (seed). The extractor function deterministically maps the inputs to an output bit stream that  has  close to uniform distribution. In designing an extractor function, one goal is to make the  length of the seed as short as possible. Extractor functions have been extensively studied in computer science literature. (Refer to \cite{Nisan:96} and \cite{NisanT:99} and the references therein for more information on this subject.)

\begin{definition}[$(\d',\e')$-extractor \cite{MaurerW:00}] Function $E:\{0,1\}^{n_1}\times\{0,1\}^{n_2}\to\{0,1\}^{n_3}$ is called $(\d',\e')$-extractor, if for any random variable $T$ with $\Tc\subseteq\{0,1\}^{n_1}$ and $H_{\infty}(T)\geq \d'n_1$,
\[
d_{\rm TV}\left([V,E(T,V)],W\right)\leq \e',
\]
where $W$ and $V$ are uniformly distributed over $\{0,1\}^{n_2+n_3}$, and$\{0,1\}^{n_2}$, respectively.
\end{definition}

The above definition has a strong requirement that if we concatenate the output $E(T,V)$ and the seed bits $V$ the whole sequence has close to uniform distribution. In the literature, weaker notions, where one is only concerned with the distribution of the output alone, has been studied as well. However, since in this paper, we are going to apply extractor functions to network coding security applications, we use the stronger notion. The seed, in our applications, is a random bit sequence that is sent from the source to all interested destinations, and is potentially observed by all eavesdroppers. Hence, we need to measure security of the code conditioned on knowing the seed.

The following lemma from \cite{Vadhan:98} shows the existence of  efficient extractor functions and gives bounds on parameters $n_2$ and $n_3$ in terms of $n_1$, $\d'$ and $\e'$.

\begin{lemma}\label{lemma:Vadhan}
Given $n_1$, $0<\d'<1$, and $\e'>0$, there exists a $(\d',\e')$-extractor $E:\{0,1\}^{n_1}\times\{0,1\}^{n_2}\to\{0,1\}^{n_3}$, such that
\[
n_2=O((\log(n_1/\e'))^2\log(\d'n_1)),
\]
and
\[
n_3=\d'n_1-2\log(1/\e')-c,
\]
where $c=O(1)$.
\end{lemma}

Clearly, for every set $\Ac_s\subset2^{\Ec}$, $\Rc_0^{(s)}\subseteq\Rc_0^{(w)}$ and $\Rc_{\e}^{(s)}\subseteq\Rc_{\e}^{(w)}$. We prove that, in case of asymptotically zero probability of error,  the other direction is true  as well, \ie  $\Rc_{\e}^{(w)}=\Rc_{\e}^{(s)}$. A similar result is proved in \cite{MaurerW:00} for wiretap channels. There, the authors show that moving from weak secrecy to strong secrecy does not change the secrecy capacity of wiretap channels. Here, we extend that result to wiretap networks. The main ingredient of our proof in  the following two lemmas from \cite{MaurerW:00}.

\begin{lemma}[Lemma 9 in \cite{MaurerW:00}]\label{lemma:8-MU}
Let $\d',\d_1,\d_2>0$. For any sufficiently large  $n_1$, there exists a function $E: \{0,1\}^{n_1}\times \{0,1\}^{n_2} \to \{0,1\}^{n_3}$, where $n_2\leq \d_1n_1$ and $n_3\geq (\d'-\d_2)n_1$, such that for every random variable  $T$ with $\Tc\subseteq\{0,1\}^{n_1}$ and $H_{\infty}(T)\geq \d'n_1$,
\[
H(E(T,V)|V)\geq n_3-2^{-\sqrt{n_1}-\d_0},
\]
where $\d_0=o(1)$, and $V$  is independent of $T$  uniformly distributed on $\{0,1\}^{n_2}$.
\end{lemma}

The proof of Lemma \ref{lemma:8-MU}, follows from Lemma \ref{lemma:Vadhan}.

\begin{lemma}[Lemma 10 in \cite{MaurerW:00}]\label{lemma:H-infty}
Let $X$ and $Y$ be two finite-alphabet random variables. For any $\l>0$,
\begin{align}
\P\left(H_{\infty}(X)-H_{\infty}(X|Y=y)\leq \log|\Yc|+\l \right)\geq 1-2^{-\l}.\label{eq:lemma:H-infty}
\end{align}
\end{lemma}

\begin{proof}
Since the proof of Lemma  \ref{lemma:H-infty} is omitted in \cite{MaurerW:00}, for completeness, we present the proof here. Note that the probability in \eqref{eq:lemma:H-infty} is  only in terms of $p(y)$, not $p(x)$. Let $x_m\triangleq \argmax_{x}p(x)$, and $x_m(y)\triangleq \argmax_xp(x|y)$.  From these definitions, we can rewrite the complement of the event mentioned in \eqref{eq:lemma:H-infty} as follows
\begin{align*}
&\P\left(H_{\infty}(X)-H_{\infty}(X|Y=y)> \log|\Yc|+\l \right)\nonumber\\
&\;=\P\left(p(x_m(Y)|Y)>p(x_m)2^{\l}|\Yc|\right).
\end{align*}
 By Markov inequality,
\begin{align}
\P&\Big(p(x_m(Y)|Y)>p(x_m)2^{\l}|\Yc|\Big)\nonumber\\
& < {\E[p(x_m(Y)|Y)]\over p(x_m)2^{\l}|\Yc|}\nonumber\\
&={ \sum_{y\in\Yc}p(y)p(x_m(y)|y)\over p(x_m)2^{\l}|\Yc|}\nonumber\\
&={ \sum_{y\in\Yc}p(x_m(y))p(y|x_m(y))\over p(x_m)2^{\l}|\Yc|}\nonumber\\
&\leq { p(x_m) \sum_{y\in\Yc}p(y|x_m(y))\over p(x_m)2^{\l}|\Yc|}\nonumber\\
&\leq { p(x_m) \sum_{y\in\Yc}1\over p(x_m)2^{\l}|\Yc|}\nonumber\\
&\leq 2^{-\l},
\end{align}
which concludes the proof.
\end{proof}

%

\section{Noiseless wiretap networks}\label{sec:wiretap}

Consider error-free communication network $\Nc$  described by acyclic directed graph $G=(\Vc,\Ec)$ such that $\Vc$ and $\Ec$ denote the set of nodes and edges in $G$, respectively.  Each node $v\in\Vc$ denotes a user in the network and each edge $e=(v_1,v_2)\in\Ec$ represents an error-free bit pipe of  finite capacity $c_e$. For each node $v\in\Vc$, let $\In(v)$ and $\Out(v)$ denote the set of  incoming and outgoing edges of node $v$, respectively. That is, $\In(v)\triangleq\{(v_1,v)\;:\:(v_1,v)\in\Ec \}$ and $\Out(v)\triangleq\{(v,v_2)\;:\:(v,v_2)\in\Ec \}$.

Let  $\Sc\subset\Vc$ and $\Tc\subset\Vc$ denote the set of source nodes and sink nodes, respectively. Without loss of generality, we  assume that:
\begin{itemize}
\item[i)] $\Sc\cap\Tc=\emptyset$.
\item[ii)] Source nodes do not have incoming edges. That is,  if $s\in\Sc$, $\In(s)=\emptyset$.
\item[iii)]  Terminal nodes do not have outgoing edges. That is,  if $t\in\Tc$, $\Out(t)=\emptyset$.
\item[iv)] Only source nodes generate keys\footnote{Note that this assumption is not restrictive.  If  node $v\in\Vc\backslash(\Vc\cup \Sc)$ desires to generate a key, we add a source node $s'$ to $\Sc$ and a directed edge $e'=(s',v)$ to $\Ec$ connecting  $s'$ to $v$, with large enough capacity $c_{e'}$. Source $s'$ does not have any message, \ie $\Mc_{s'}=\emptyset$, but it generates key $K_{s'}$, and sends it directly to $v'$ through $e'$.}.
\end{itemize}

A  secure network code of blocklength $n$ operating on network $\Nc$ is defined  as follows. Each source $s\in\Sc$ observes  message $M_s$ and  key $K_s$  uniformly distributed over
\[
\Mc_{s}\triangleq[1:2^{\lfloor n r_{m_s}\rfloor}],
\]
and
\[
\Kc_s\triangleq[1:2^{\lfloor n r_{k_s}\rfloor}],
\]
 respectively. Here, $r_{m_s}$ and $r_{k_s}$ represent the information rate and the key rate of source $s$, respectively.

 Let $W_e\in\Wc_e$ denote the message traversing edge $e\in\Ec$. While  most of the network coding literature is dedicated to fixed-length  coding across channels, in this paper, we consider both variable-length and fixed-length channel coding models, described as follows.
 \begin{itemize}
 \item[i)] Fixed-length channel coding: In the fixed-length model,   a bit pipe of capacity $c_e$ is defined  to be an error-free point-to-point channel, which carries at most $nc_e$ bits in $n$ transmissions. Therefore, in this case the rate of the code used on  channel $e$ is defined as
\begin{align}\label{ref:cap-const}
r_e^{(f)}\triangleq n^{-1} \log|\Wc_e|,
 \end{align}
 or equivalently we can assume that\footnote{Throughout the paper the basis of logarithm is assumed to be 2.} $\Wc_e=[1:2^{nr_e^{(f)}}]$.

  \item[ii)] Variable-length channel coding:  In the variable-length coding regime, for a blocklength of $n$,  the number of bits sent over channel $e$ with capacity $c_e$ is allowed to exceed $nc_e$. Therefore, in the case of variable-length coding, the rate of the  code, $r_e^{(v)}$, is measured as the average \emph{expected} number of bits per channel use, or, more precisely,
\begin{align}\label{ref:cap-const}
r_e^{(v)}\triangleq n^{-1} H(W_e).
 \end{align}
From the definition of $r_e^{(v)}$,  unlike the case of fixed-length coding, in this case $n^{-1}\log |\Wc_e|$ can be much larger than $c_e$.

 \end{itemize}

As mentioned earlier, most of the literature on analyzing  network coding capacity regions is dedicated to the case of asymptotically zero probability of error where the channel code used on each point-point  channel is a fixed-length code. The reason in this paper we consider both fixed and variable length coding is that, here we focus on both asymptotically zero and perfectly zero probability of error. While fixed-length coding is a  proper model for analyzing the capacity region of the former case, in the latter one, in order to avoid  shrinking  the achievable capacity region by a significant amount, we move to the variable-length coding paradigm.   The advantage of this shift is that when using typical sets as the building blocks,  variable-lengths codes allow us to once in a while send a non-typical sequence,  without hurting the performance significantly.

 In \cite{SongY:03}, where the authors derive inner and outer bounds on the zero-error capacity region of acyclic networks (with no security constraint), there too the authors consider variable-length coding across the channels. However, at this point it is not clear to us, whether this shift of models is essential to the problem, or  is  a technical issue that can be fixed later.

For each edge $e=(v_1,v_2)\in\Ec$, the function $\eta_{e}$ denotes the encoding operation performed by node $v_1$ to generate the message sent over edge $e$. If $e\in\Out(s)$, for some $s\in\Sc$, then
\[
\eta_e:\;\Mc_{s}\times\Kc_s\;\to\; \Wc_e,
\]

and
\[
W_e=\eta_e(M_s,K_s).
\]  If $e\in\Out(v)$, where $v\in\Vc\backslash(\Sc\cup\Tc)$, then
\[
\eta_e:\; \prod_{e'\in\In(v)}\Wc_{e'}\;\to\; \Wc_e,
\]
and $W_e=\eta_e(W_{e'}: \;e'\in\In(v))$. Each terminal node $t\in\Tc$ reconstructs a subset $\b(t)\subseteq \Sc$ of source messages. For $s\in\b(t)$, let $\eta_{t,s}$ denote the decoding function applied at terminal node $t$ for decoding source $s$. More precisely,
\[
\eta_{t,s}:\;  \prod_{e\in\In(t)}\Wc_e \;\to\; \Mc_s,
\]
and
\[
\hat{M}_{s\to t}=\eta_{t,s}(W_e:\; e\in\In(t)),
\]
where $\hat{M}_{s\to t}$ represents the reconstruction of message $M_s$ at node $t$.

Let $\Ac_s\subset2^{\Ec}$ denote the set of security constraints described as follows.  Each member $\a\in\Ac_s$ represents a subset of edges $\Ec$. An eavesdropper can choose a   set $\a\in\Ac_s$, and access the information transmitted on all the edges  $e\in\a$.  Our goal is to prevent the leakage of information to the  eavesdropper such that  asymptotically it does not get any meaningful information about the source messages. In the following,  this notion is defined more precisely.

\begin{enumerate}
\item Zero-error capacity region: Rate  $\rv=(r_s:s\in\Sc)$ is said to be achievable  subject to weak secrecy with respect to the set $\Ac_s\subset2^{\Ec}$ at zero error rate,  if for any $\e>0$, there exists  a code of blocklength $n$ large enough, message rates $\rv_m^{(n)}=(r^{(n)}_{m_s}:s\in\Sc)$ and key rates  $\rv^{(n)}_k=(r^{(n)}_{k_s}:s\in\Sc)$, such that
\begin{enumerate}
\item[i.] For every $s\in\Sc$,
\[r_{m_s}^{(n)} \geq r_s-\e.\]
\item[ii.] For every $e\in\Ec$, \[r^{(v)}_e\leq c_e.\]
\item[iii.]  For every $\a\in\Ac_s$,
\[n^{-1} I(M_{\Sc};W_{\a}) \leq \e,\]
where $M_{\Sc}=(M_s:\; s\in\Sc)$ and $W_{\a}=(W_e:\:e\in\a)$.
\item[iv.] For every $t\in\Tc$, and $s\in\b(t)$,
\[
\P(\hat{M}_{s\to t}\neq M_{s}) = 0.
\]
 \end{enumerate}
 Rate   $\rv=(r_s:s\in\Sc)$ is said to be  achievable  subject to strong secrecy with respect to the set $\Ac_s\subset2^{\Ec}$ at zero error rate, if it satisfies identical conditions as those required for weakly-secure achievable rates, except for condition (iii), where in this case we require
\[ I(M_{\Sc};W_{\a}) \leq \e,\]
for every $\a\in\Ac_s$.

Let $\Rc^{(w)}_0$ denote the set of all achievable weakly-secure rates with zero error. Similarly, let $\Rc^{(s)}_0$ denote the set of all achievable rates at zero error subject to strong security.

\item Asymptotically zero-error capacity region: Rate  $\rv=(r_s:s\in\Sc)$ is said to be  achievable  at asymptotically zero error subject to weak secrecy with respect to the set $\Ac_s\subset2^{\Ec}$,  if for any $\e>0$, there exists  a code with blocklength $n$ large enough with message rates $\rv_m^{(n)}=(r^{(n)}_{m_s}:s\in\Sc)$  key rates  $\rv^{(n)}_k=(r^{(n)}_{k_s}:s\in\Sc)$, such that
\begin{enumerate}
\item[i.] For every $s\in\Sc$,
\[r_{m_s}^{(n)} \geq r_s-\e.\]
\item[ii.] For every $e\in\Ec$, \[r^{(f)}_e\leq c_e.\]
\item[iii.]  For every $\a\in\Ac_s$,
\[n^{-1} I(M_{\Sc};W_{\a}) \leq \e,\]
where $M_{\Sc}=(M_s:\; s\in\Sc)$ and $W_{\a}=(W_e:\:e\in\a)$.
\item[iv.] For every $t\in\Tc$, and $s\in\b(t)$,
\[
\P(\hat{M}_{s\to t}\neq M_{s}) \leq \e.
\]
 \end{enumerate}

 Again, rate  $\rv=(r_s:s\in\Sc)$ is said to be  achievable  subject to strong secrecy with respect to the set $\Ac_s\subset2^{\Ec}$ at asymptotically zero error rate, if all the above conditions hold, except for condition (iii), which is strengthened in this case as
\[I(M_{\Sc};W_{\a}) \leq \e,\]
for every $\a\in\Ac_s$.

Let   $\Rc_{\e}^{(w)}$ and $\Rc_{\e}^{(s)}$ denote the set of all weakly-secure and strongly-secure achievable rates at asymptotically zero error, respectively. \end{enumerate}

\begin{remark}
We define the capacity constraint as $r^{(v)}_e\leq c_e$ and $r^{(f)}_e\leq c_e$, in the case of zero error and asymptotically zero error, respectively.  This is stronger than the conditions  usually considered in the network coding literature, \ie $r^{(v)}_e\leq c_e+\e$ or $r^{(f)}_e\leq c_e+\e$, where $\e$ can be made arbitrarily small \cite{yeung}. We believe that the stricter assumptions are more natural, because noiseless links of finite capacities are commonly used to model noisy channels of equal capacities \cite{KoetterE:11,JalaliE:10,JalaliE:11}. However, a noisy channel of capacity $c_e$, at blocklength $n$, can  at most carry $nc_e$ bits with arbitrarily  small probability of error, which is equivalent to requiring $r^{(v)}_e\leq c_e$ and $r^{(f)}_e\leq c_e$.
\end{remark}



Note that  noiseless wiretap networks are  almost a  special case of  DMNs. The only  difference between the two models  is that in DMNs we assume that the adversary has access to the outputs of some nodes in the network, while in wiretap networks, the eavesdropper listens to some links. However, the proof of Theorem \ref{thm:strong_eq_weak} can readily  be extended    to noiseless  wiretap networks as well.   The result is summarized in Corollary \ref{cor:strong-weak-epsilon} of Theorem \ref{thm:strong_eq_weak}:

\begin{corollary}\label{cor:strong-weak-epsilon}
For any wiretap network $\Nc$ described by graph $G=(\Vc,\Ec)$, set of sources $\Sc$, set of terminals $\Tc$, $(\b(t): t\in\Tc)$, and    collection of  potential sets of links $\Ac_s\subset2^{\Ec}$,
\[
\Rc_{\e}^{(w)}=\Rc_{\e}^{(s)}.
\]
\end{corollary}


\section{Capacity region: zero-error coding}\label{sec:c-region-0}

In this section we derive inner and outer bounds on the secure  network coding capacity region subject to weak secrecy constraint and zero probability of error. As a reminder,  in the case of zero-error communication, we  consider  variable-length coding across the links.

Consider a set of finite-alphabet random variables $((U_{m_s},U_{k_s})_{s\in\Sc},(U_e)_{e\in\Ec})$ corresponding to the source messages, source keys and codewords sent over the edges. Define $N\triangleq 2|\Sc|+|\Ec|$  to denote the total number of  such random variables in the network. As defined earlier, $\Gamma_{N}^*$ denotes the entropic region defined by all possible set of  $N$ finite-alphabet random variables.

Define sets $\Gamma_1,\ldots,\Gamma_6$ as follows:
\begin{align}
\Gamma_1&\triangleq \big\{\hv\in\Hc_{N}:\; h_{\Mc_s\cup\Kc_s}=\sum_{s\in\Sc}(h_{m_s}+h_{k_s})\big\},\\
\Gamma_2&\triangleq \big\{\hv\in\Hc_{N}:\; h_{m_s,k_s,\Out(s)}=h_{m_s,k_s}, \;\forall \;s\in\Sc \},\\
\Gamma_3&\triangleq\big \{\hv\in\Hc_{N}:\; h_{\In(v)\cup\Out(v)}=h_{\In(v)}, \;\forall\;v\in\Vc\backslash\Sc\cup\Tc \},\\
\Gamma_4&\triangleq \big\{\hv\in\Hc_{N}:\; h_{m_{\b(t)},\In(t)}= h_{\In(t)}, \;\forall\;t\in\Tc \},\\
\Gamma_5&\triangleq \big\{\hv\in\Hc_{N}:\; h_{e}\leq c_e, \;\forall\; e\in\Ec \},\\
\Gamma_6&\triangleq \big\{\hv\in\Hc_{N}:\; h_{m_{\Sc},\a}= h_{m_{\Sc}}+h_{\a}, \; \forall\a\in\Ac_s \}.
\end{align}

For $\b\subseteq[1:6]$, let
\[
\Gamma_{\b}\triangleq \bigcap_{i\in\b} \Gamma_i.
\]
Also, for notational simplicity, sometimes instead of $\Gamma_{\{i,j\}}$, we write $\Gamma_{ij}$.

Define the projection operator ${\rm Proj}_{\Sc}(\cdot)$ over the set  $\Hc_{N}$ as follows: for any vector $\hv\in\Hc_N$,  $\bf{f}={\rm Proj}_{\Sc}(\hv)\in\mathds{R}^{|\Sc|}$ and $f_{s}=h_s$, for all $s\in\Sc$. In other words, ${\rm Proj}_{\Sc}(\cdot)$ projects a vector in $\Hc_N$ onto the coordinates corresponding to the source messages.

\begin{theorem} [Inner bound]\label{thm:inner_bd-variable}
Define
\[
\Rc_{0,i}\triangleq \Lambda({\rm Proj}_{\Sc}(\overline{\rm con}(\Gamma_{N}^{*}\cap \Gamma_{12346})\cap\Gamma_5)).
\]
Then,
$\Rc_{0,i}\subseteq \Rc_0^{(w)}.$
\end{theorem}

\begin{theorem} [Outer bound]\label{thm:outer_bd-variable}
Define
\[
\Rc_{0,o} \triangleq\Lambda( {\rm Proj}_{\Sc}(\overline{\rm con}(\Gamma_{N}^{*}\cap \Gamma_{1234})\cap\Gamma_{56})).
\]
Then, $\Rc_0^{(w)}\subseteq \Rc_{0,o}$.
\end{theorem}

Before proving Theorem \ref{thm:inner_bd-variable}, note that since $\Gamma_{12346}$ is a linear subspace of $\Hc_{N}$, an straightforward extension of  Lemma 21.8 in \cite{yeung} yields:

\begin{lemma}\label{lemma:1}
\[
\overline{\rm con}(\Gamma_{N}^{*}\cap \Gamma_{12346}) = \overline{D(\Gamma_{N}^{*}\cap \Gamma_{12346})}
\]
\end{lemma}
This result is used in the proof of the results.

\subsection{Proof of Theorem \ref{thm:inner_bd-variable}}

If $\rv\in\Rc_0^{(w)}$, then by our definition of achievability, any rate vector $\rv'\leq \rv$ is also achievable. Hence, to prove the theorem, it is enough to show that $\Rc_{0,i}'\triangleq{\rm Proj}_{\Sc}(\overline{\rm con}(\Gamma_{N}^{*}\cap \Gamma_{12346})\cap\Gamma_5)\subseteq \Rc$.

Let $\rv\in\Rc'_{0,i}$. By the definition of $\Rc'_{0,i}$ and Lemma \ref{lemma:1},  there exists  sequences $\hv^{(k)}\in \Gamma_{N}^{*}\cap \Gamma_{12346}$,  $a_k\in(0,1]$, and $\e_k\in(0,1]$ such that, for any $s\in\Sc$,
\begin{align}
a_k h^{(k)}_{m_s} \geq r_s-\e_k,\label{eq:source rate constraint}
\end{align}
for each edge $e\in\Ec$,
\begin{align}
a_kh^{(k)}_e\leq c_e+\e_k,\label{eq:ce}
\end{align}
 and $\e_k\to 0$, as $k$ grows to infinity.

Based on the sequence of entropic vectors  $\hv^{(k)}$, we build a sequence of weakly secure network codes with zero probability of error that asymptotically achieves rate $\rv$.

Since $\hv^{(k)}\in\Gamma_{N}^{*}\cap \Gamma_{12346}$, there exists a set of finite-alphabet random variables random variables, $\{U_{m_s},U_{k_s}: s\in\Sc\}\cup\{U_e\; : \: e\in\Ec\}$, which satisfy the following:
\begin{itemize}
\item[i.]   Random variables $(U_{m_s},U_{k_s}: s\in\Sc)$ are all independent. That is,
\begin{align}\label{eq:i}
H(U_{m_s},U_{k_s}: s\in\Sc)=\sum_{s\in\Sc} H(U_{m_s}) + \sum_{s\in\Sc}H(U_{k_s}).
\end{align}

\item[ii.] For each outgoing edge $e\in\Out(s)$ of a source node $s\in\Sc$, since $h_{m_s,k_s,\Out(s)}=h_{m_s,k_s}$, there exists a deterministic function $\eta_e$,
\begin{align}\label{eq:ii}
\eta_e:\Uc_{m_s}\times\Uc_{k_s}\to\Uc_e,
\end{align}
 such that $U_e=\eta_e(U_{m_s},U_{k_s})$.

\item[iii.] Similarly, for each edge $e=(v_1,v_2)$ with $v_1\notin\Sc$, since $h_{\In(v_1),\Out(v_1)}=h_{\In(v_1)}$, there exists a deterministic function $\eta_e$,
\begin{align}\label{eq:iii}
\eta_e:\;\prod_{e'\in\In(v_1)}\;\Uc_{e'}\to \Uc_e,
\end{align}
such that $U_e=\eta_e(U_{e'}:\; e'\in\In(v_1))$.

\item[iv.]  Corresponding to each sink node $t\in\Tc$ and source node $s\in\b(t)$, there exists a  decoding function $\eta_{t,s}$,
\begin{align}\label{eq:vi}
\eta_{t,s}:\;\prod_{e\in\In(t)}\;\Uc_e\to \Uc_{m_s},
\end{align}
 such that $U_{m_s}=\eta_{t,s}(U_e:\;e\in\In(t))$.

\end{itemize}

\subsubsection{Random code construction}

For each source $s$, consider the sets of $\e$-typical sequences $\Tc_{\e}^{(n_t)}(U_{m_s})$ and $\Tc_{\e}^{(n_t)}(U_{k_s})$ and let
\[
\Mc_s\triangleq \Tc_{\e}^{(n_t)}(U_{m_s}),
\]
and
\[
\Kc_s\triangleq \Tc_{\e}^{(n_t)}(U_{k_s}).
\]
 Moreover, define
 \[
 |\Mc_s|\triangleq2^{nr^{(k)}_{m_s}},
 \]
 and
 \[|
 \Kc_s|\triangleq2^{nr^{(k)}_{k_s}},
 \]
 where $n$ denotes the blocklength of the code and will be specified later.

 For  edge $e\in\Ec$, define random variable $I^{(n_t)}_e$ as an indicator function of whether the message traversing edge $e$ is a typical sequence, \ie
\[
I^{(n_t)}_e \triangleq \ind_{U_e^{n_t}\in\Tc_{\e}^{(n_t)}(U_e)}.
\]

The blocklength $n$ is defined as
\begin{align}
n \triangleq \left\lceil{n_t(1+\delta)\over a_k}\right\rceil,\label{eq:n}
\end{align}
where
\begin{align}
\delta \triangleq 2\Big(\max_{e\in\Ec}\Big((1+\e)\Big(1+{\e_k\over c_e}\Big)+{a_k\log|\Uc_e|\P(I^{(n_t)}_e=0)\over c_e}\Big)-1\Big).\label{eq:delta}
\end{align}
From this definition, for fixed $\e_k$, choosing $n_t$ large enough and $\e$ small enough, $\d$ can be made arbitrarily small.

%

After building the message and key codebooks in this manner, the encoding operations are performed as follows.  Each source $s\in\Sc$ chooses message $M_s$ and key $K_s$ uniformly at random from sets $\Mc_s$ and $\Kc_s$, respectively. Corresponding to  each edge $e\in\Out(s)$, by applying the function $\eta_e$ symbol-by-symbol to $(U_{m_s}^{n_t},U_{k_s}^{n_t})$, the source constructs codeword $W_e=U_e^{n_t}$. In other words,  for $e\in\Out(s)$, and $i=1,\ldots, n_t$,
\begin{align}
U_{e,i}=\eta_e(U_{m_s,i},U_{k_s,i}).\label{eq:coding-op1}
\end{align}

At  node $v\in\Vc\backslash(\Sc\cup\Tc)$, the codeword $W_e$ sent over edge $e\in\Out(v)$ is constructed from $(W_{e'})_{e'\in\In(v)}=(U^{n_t}_{e'})_{e'\in\In(v)}$ by applying the function $\eta_e$ symbol-by-symbol. That is, for $i\in[1:n_t]$,
\begin{align}
U_{e,i}=\eta_e(U_{e',i}:\; e'\in\In(v)).\label{eq:coding-op2}
\end{align}

Finally, sink node $t\in\Tc$ reconstructs the message at source $s\in\b(t)$, $\hat{M}_{s\to t}=\hat{U}_{m_s}^{n_t}$, as a function of $(W_e)_{e\in\In(t)}=(U_e^{n_t})_{e\in\In(t)}$,  such that, for $i=1,\ldots,n_t$,
\[
\Uh_{m_s,i}=\eta_{t,s}(U_{e,i}:\;e\in\In(t)).
\]

\subsubsection{Performance evaluation}

\begin{itemize}
\item[i.] Source rates:
From \eqref{eq:size-typical-set}, for $n_t$  large enough,
 \begin{align}
(1-\e) 2^{n_t(1-\e)h^{(k)}_{m_s}}< | \Tc_{\e}^{(n_t)}(U_{m_s})| < 2^{n_t(1+\e)h^{(k)}_{m_s}},\label{eq:size-Ms}
 \end{align}
 where  by our definition $h^{(k)}_{m_s}=H(U_{m_s})$. Hence, by \eqref{eq:source rate constraint} and \eqref{eq:n},
 \begin{align}
r^{(k)}_{m_s}&={\log| \Tc_{\e}^{(n_t)}(U_{m_s})|\over n}\nonumber\\
& \geq {n_t\over n}(1-\e)h^{(k)}_{m_s}  +{\log (1-\e)\over n}\nonumber\\
&  \geq {n_t(1-\e)\over n_t(1+\delta)+a_k}(r_s-\e_k)+ {a_k\log (1-\e)\over n_t(1+\delta)+a_k}.\label{eq:lower-bd-rs}
 \end{align}
Choosing $\e_k$ small enough, and then $\e$ and $n_t$ appropriately, we can make the right hand side of \eqref{eq:lower-bd-rs} arbitrarily close to $r_s$, for all $s\in\Sc$.

\item[ii.] Channels capacity constraints: As mentioned earlier, in the case of zero-error decoding, we consider variable-length channel coding. Hence, to make sure that the capacity constraints are satisfied, we need to prove that  $H(W_e)\leq nc_e$, for all $e\in\Ec$. To prove this, note that
\begin{align}
H(W_e) &\leq H(W_e,I_e)\nonumber\\
&=  H(U_e^{n_t},I_e)\nonumber\\
&= H(I_e) + H(U_e^{n_t}|I_e)\nonumber\\
& \leq 1 +H(U_e^{n_t}|I_e=1)\P(I_e=1)\nonumber\\
&\;\;\;\;+H(U_e^{n_t}|I_e=0)\P(I_e=0)\nonumber\\
& \leq 1 +\log|\Tc_{\e}^{(n_t)}(U_e)|+H(U_e^{n_t}|I_e=0)\P(I_e=0),\label{eq:upper bd on H(We)}
 \end{align}
 where, for simplifying the notation, we have replaced $I^{(n_t)}_e$ in \eqref{eq:upper bd on H(We)} by $I_e$.
Again, for $n_t$ large enough,
 \begin{align}\label{eq:size of Ue}
(1-\e) 2^{n_t(1-\e)H(U_e)}< | \Tc_{\e}^{(n_t)}(U_e^{n_t})| < 2^{n_t(1+\e)H(U_e)}.
 \end{align}
Combining \eqref{eq:ce}, \eqref{eq:n},  \eqref{eq:delta}, \eqref{eq:upper bd on H(We)} and \eqref{eq:size of Ue} yields
\begin{align}
H&(W_e) \leq 1 +n_t(1+\e)H(U_e)+n_t\log|\Uc_e|\P(I_e=0)\nonumber\\
&\leq 1 +{na_k\over 1+\d}((1+\e)h_{e}^{(k)}+\log|\Uc_e|\P(I_e=0))\nonumber\\
&\leq 1 +{na_k\over 1+\d}((1+\e)({c_e+\e_k\over a_k})+\log|\Uc_e|\P(I_e=0))\nonumber\\
&\leq 1 +{nc_e\over 1+\d}\Big((1+\e)(1+{\e_k\over c_e})+{a_k\log|\Uc_e|\P(I_e=0)\over c_e}\Big)\nonumber\\
&\leq 1+\left({1+0.5\d\over 1+\d}\right) nc_e\nonumber\\
&\leq nc_e,
 \end{align}
 where the last step holds only for $n$ large enough.

\item[iii.] Decoding error probability: We prove that  the decoding error probability of the described code, \ie
\[
P_e^{(n)}\triangleq \P(\hat{M}_{s\to t}\neq M_{s}, {\rm for}\;{\rm some}\;t\in\Tc,\;s\in\b(t)),
\]
 is zero. To prove this, we show that $\P(\hat{M}_{s\to t}\neq M_s)=0$, for any $t\in\Tc$ and $s\in\b(t)$.

As mentioned earlier, since $\hv^{(k)}\in\Gamma_4$, the jointly distributed random variables $(U_e)_{e\in\Ec}$ and $(U_{m_s},U_{k_s})_{s\in\Sc}$ satisfy
\[
U_{m_s}=\eta_{t,s}(U_e:\;e\in\In(t)),
\]
for all $t\in\Tc$ and $s\in\b(t)$, where $\eta_{t,s}(\cdot)$ is a deterministic function.  This implies that for  any  $u_{m_s}\in\Uc_{m_s}$, $t\in\Tc$ and $u_{\In(t)}\in\prod_{e\in\In(t)}\Uc_e$,
\begin{align}
\P(U_{m_s}=u_{m_s}|U_{\In(t)}=u_{\In(t)})=1,
\end{align}
if  $u_{m_s}=\eta_{t,s}(u_{\In(t)})$, and zero otherwise. Assume that for $u_{m_s}\in\Uc_s$ and  $u_{\In(t)}\in\prod_{e\in\In(t)}\Uc_e$ are such that $\P(U_{k_s}=u_{m_s}|U_{\In(t)}=u_{\In(t)})=0$. Then, it follows that
\begin{align}
0&=\sum\limits_{\substack{(u_{m_{\Sc\backslash\{s\}}},u_{k_{\Sc}})\in \\\prod\limits_{s'\in\Sc\backslash\{s\}}\Uc_{m_{s'}}\times \prod\limits_{s'\in\Sc}\Uc_{k_{s'}}}} p(u_{m_s},u_{m_{\Sc\backslash\{s\}}},u_{k_{\Sc}}|u_{\In(t)})\nonumber\\
&=\sum\limits_{(u_{m_{\Sc\backslash\{s\}}}, u_{k_{\Sc}})} p(u_{m_{\Sc\backslash\{s\}}},u_{k_{\Sc}}|u_{\In(t)}) p(u_{m_s}|u_{\In(t)},u_{m_{\Sc\backslash\{s\}}},u_{k_{\Sc}}).\label{eq:zero-prob}
\end{align}
Therefore, from \eqref{eq:zero-prob}, for any $u_{m_{\Sc\backslash\{s\}}}\in\prod_{s'\in\Sc\backslash\{s\}}\Uc_{s'}$ and $u_{k_{\Sc}}\in \prod_{s'\in\Sc}\Uc_{k_{s'}}$ such that $p(u_{m_{\Sc\backslash\{s\}}},u_{k_{\Sc}}|u_{\In(t)})\neq 0$, or equivalently $p(u_{m_{\Sc\backslash\{s\}}},u_{k_{\Sc}}|u_{\In(t)})\neq 0$,
\begin{align}
p(u_{m_s}|u_{\In(t)},u_{m_{\Sc\backslash\{s\}}},u_{k_{\Sc}})=0.\label{eq:reconstructions}
\end{align}
Equation \eqref{eq:reconstructions} implies that for any  $u_{m_{\Sc\backslash\{s\}}}$, $u_{k_{\Sc}}$ and $u_{\In(t)}$ such that there exists $u_{m_s}\in\Uc_{m_s}$, for which applying encoding operations $\eta_e(\cdot)$ to input vectors $u_{m_{\Sc}}$, $u_{k_{\Sc}}$  results in the input edges of sink node $t$ taking values  $u_{\In(t)}$, we have
\[
u_{m_s} = \eta_{t,s}(u_{\In(t)}).
\]
Clearly, by the union bound,
\begin{align}
\P(\Uh_{m_s}^{n_t}\neq U_{m_s}^{n_t}) \leq \sum_{i=1}^{n_t}\P(\Uh_{m_s,i}\neq U_{m_s,i}).
\end{align}
Hence to show that $\P(\Uh_{m_s}^{n_t}\neq U_{m_s}^{n_t})=0$, it suffices to show that $\P(\Uh_{m_s,i}\neq U_{m_s,i})=0$, for all $i=1,\ldots,n_t$. But, the above argument guarantees that
\[
\P(\Uh_{m_s,i}\neq U_{m_s,i})=0,
\]
for all $i=1,\ldots,n$.

\item[iv.] Security guarantee: We  prove that for any $\a\in\Ac_s$,
\[
{1\over n} I(M_{\Sc};W_{\a})={1\over n} I(U^{n_t}_{m_{\Sc}};U_{\a}^{n_t})
\]
can be made arbitrarily small.

 For any $\a\in\Ac_s$ and $u_{m_{\Sc}}^{n_t}\in\prod_{s\in\Sc}\Mc_s$, we have
\begin{align}
\P(U_{\a}^{n_t}=u_{\a}^{n_t}|U_{m_{\Sc}}^{n_t}=u_{m_{\Sc}}^{n_t})&=\sum_{u_{k_{\Sc}}^{n_t}\in\Uc_{k_{\Sc}}^{n_t}}\P(U_{\a}^{n_t}=u_{\a}^{n_t},U_{k_{\Sc}}^{n_t}=u_{k_{\Sc}}^{n_t}|U_{m_{\Sc}}^{n_t}=u_{m_{\Sc}}^{n_t})\nonumber\\
&=\sum_{u_{k_{\Sc}}^{n_t}\in\Uc_{k_{\Sc}}^{n_t}}\P(U_{\a}^{n_t}=u_{\a}^{n_t}|U_{k_{\Sc}}^{n_t}=u_{k_{\Sc}}^{n_t},U_{m_{\Sc}}^{n_t}=u_{m_{\Sc}}^{n_t})\P(U_{k_{\Sc}}^{n_t}=u_{k_{\Sc}}^{n_t})\nonumber\\
&=\sum_{u_{k_{\Sc}}^{n_t}\in\Uc_{k_{\Sc}}^{n_t}}\;p_{\Kc_{\Sc}}(u_{k_{\Sc}}^{n_t})\;\prod_{i=1}^{n_t}\;p_{U_{\a}|U_{k_{\Sc},m_{\Sc}}}(u_{\a,i}|u_{k_{\Sc},i},u_{m_{\Sc},i}),\label{eq:secrecy-step1}
\end{align}
where the last step follows from our symbol-by-symbol  coding strategy described by  \eqref{eq:coding-op1} and \eqref{eq:coding-op2}.

Let \[\Kc_{\Sc}\triangleq \prod_{s\in\Sc}\Kc_s.\] At source node $s\in\Sc$, the message and the key  are drawn from a uniform distribution from $\Mc_s$ and $\Kc_s$, respectively. Hence, for $u_{k_{\Sc}}^{n_t}\in\Kc_{\Sc}$,
\[
\P(U_{k_{\Sc}}^{n_t}=u_{k_{\Sc}}^{n_t}) = \prod_{s\in\Sc} {1\over |\Kc_s|} = \prod_{s\in\Sc} 2^{-nr_{k_s}^{(k)}}=  2^{-n\sum\limits_{s\in\Sc}r_{k_s}^{(k)}}.
\]
For $u_{k_{\Sc}}^{n_t}\in\prod_{s\in\Sc}\Uc^{n_t}_{m_s} \backslash \Kc_{\Sc}$,
\[
p_{\Kc_{\Sc}}(u_{k_{\Sc}}^{n_t}) = 0.
\]
For $u_{k_s}^{n_t}\in\Kc_s$, since $\Kc_s= \Tc_{\e}^{(n_t)}(U_{k_s})$,
\begin{align}
2^{-n_th_{k_s}^{(k)}(1+\e)}\leq p_{U_{k_s}}(u_{k_s}^{n_t})\leq2^{-n_th_{k_s}^{(k)}(1-\e)},\label{eq:prob uk iid}
\end{align}
where $p_{U_{k_s}}(u_{k_s}^{n_t})=\prod_{i=1}^{n_t}p_{U_{k_s}}(u_{k_s,i})$. On the other hand,
\begin{align}
p_{\Kc_s}(u_{k_s}^{n_t})&={1\over|\Kc_s|}=2^{-nr_{k_s}^{(k)}}.\label{eq:prob uk codebook}
\end{align}
Combining \eqref{eq:size-typical-set}, \eqref{eq:prob uk iid} and \eqref{eq:prob uk codebook} yields
\begin{align}
2^{-2n_t\e h_{k_s}^{(k)}}p_{U_{k_s}}(u_{k_s}^{n_t}) \leq p_{\Kc_s}(u_{k_s}^{n_t})\leq({1\over 1-\e}) 2^{2n_t\e h_{k_s}^{(k)}}p_{U_{k_s}}(u_{k_s}^{n_t}),\label{eq:bound-prob-in-Ks}
\end{align}
for any $u_{k_s}^{n_t}\in\Kc_s$.
For $u_{k_s}^{n_t}\notin\Kc_s$, $p_{\Kc_s}(u_{k_s}^{n_t})=0$, and hence
\begin{align}
0 \leq p_{\Kc_s}(u_{k_s}^{n_t})\leq p_{U_{k_s}}(u_{k_s}^{n_t}).\label{eq:bound-prob-notin-Ks}
\end{align}
Therefore, by inserting \eqref{eq:bound-prob-in-Ks}, and \eqref{eq:bound-prob-notin-Ks} in \eqref{eq:secrecy-step1}, it follows that
\begin{align}
\P&(U_{\a}^{n_t}=u_{\a}^{n_t}|U_{m_{\Sc}}^{n_t}=u_{m_{\Sc}}^{n_t})\nonumber\\
&=\sum_{u_{k_{\Sc}}^{n_t}\in\Kc_{\Sc}}\;p_{\Kc_{\Sc}}(u_{k_{\Sc}}^{n_t})\;\prod_{i=1}^{n_t}\;p_{U_{\a}|U_{k_{\Sc},m_{\Sc}}}(u_{\a,i}|u_{k_{\Sc},i},u_{m_{\Sc},i})\nonumber\\
&\stackrel{(a)}{\leq} \sum_{u_{k_{\Sc}}^{n_t}\in\Uc_{k_{\Sc}}^{n_t}}({1\over 1-\e})^{|\Sc|} 2^{2n_t\e (\sum\limits_{s\in\Sc}h_{k_s}^{(k)})}\;\prod_{i=1}^{n_t}\;p_{U_{\a}|U_{k_{\Sc},m_{\Sc}}}(u_{\a,i}|u_{k_{\Sc},i},u_{m_{\Sc},i})p_{U_{k_s}}(u_{k_s,i})\nonumber\\
&\leq({1\over 1-\e})^{|\Sc|} 2^{2n_t\e (\sum\limits_{s\in\Sc}h_{k_s}^{(k)})} \;\prod_{i=1}^{n_t} \sum_{u_{k_{\Sc,i}}\in\Uc_{k_{\Sc}}}p_{U_{\a}|U_{k_{\Sc},m_{\Sc}}}(u_{\a,i}|u_{k_{\Sc},i},u_{m_{\Sc},i})\prod_{s\in\Sc}p_{U_{k_s}}(u_{k_s,i}),\label{eq:secrecy-step3}
\end{align}
where (a) follows from the fact that $p_{\Kc_{\Sc}}(u_{k_{\Sc}}^{n_t})=0$, for $u_{k_{\Sc}}^{n_t}\in\Uc_{k_{\Sc}}^{n_t}\backslash\Kc_{\Sc}$.

By  assumption, $\hv^{(k)}\in\Gamma_6$,  and therefore, $I(U_{\a};U_{m_{\Sc}})=0$. This implies that, for any $u_{\a}\in\Uc_{\a}$ and $u_{m_{\Sc}}\in\Uc_{m_{\Sc}}$, $p(u_{\a}|u_{m_{\Sc}})=p(u_{\a})$. Hence,
\begin{align}
\sum_{u_{k_{\Sc,i}}} p_{U_{\a}|U_{k_{\Sc},m_{\Sc}}}(u_{\a,i}|u_{k_{\Sc},i},u_{m_{\Sc},i})\prod_{s\in\Sc}p_{U_{k_s}}(u_{k_s,i})
&=\sum_{u_{k_{\Sc,i}}} p_{U_{\a},U_{k_{\Sc}}|U_{m_{\Sc}}}(u_{\a,i},u_{k_{\Sc},i}|u_{m_{\Sc},i})\nonumber\\
&= p_{U_{\a}|U_{m_{\Sc}}}(u_{\a,i}|u_{m_{\Sc},i})\nonumber\\
&= p_{U_{\a}}(u_{\a,i}).\label{eq:independence of alpha from m}
\end{align}
Combining \eqref{eq:secrecy-step3} and \eqref{eq:independence of alpha from m}, it follows that
\begin{align}
\P(U_{\a}^{n_t}=u_{\a}^{n_t}|U_{m_{\Sc}}^{n_t}=u_{m_{\Sc}}^{n_t})&\leq({1\over 1-\e})^{|\Sc|} 2^{2n_t\e (\sum\limits_{s\in\Sc}h_{k_s}^{(k)})} \;\prod_{i=1}^{n_t} p_{U_{\a}}(u_{\a,i})\nonumber\\
&=({1\over 1-\e})^{|\Sc|} 2^{2n_t\e (\sum\limits_{s\in\Sc}h_{k_s}^{(k)})}  p_{U_{\a}}(u_{\a}^{n_t}),
\end{align}
or equivalently,
\begin{align}
-{1\over n}\log \P(U_{\a}^{n_t}=u_{\a}^{n_t}|U_{m_{\Sc}}^{n_t}=u_{m_{\Sc}}^{n_t})&\geq -{1\over n}\log p_{U_{\a}}(u_{\a}^{n_t}) + {|\Sc|\log( 1-\e)\over n} -  2\e ({n_t\over n})(\sum\limits_{s\in\Sc}h_{k_s}^{(k)}).\label{eq:lb cond prob}
\end{align}
Multiplying both sides of \eqref{eq:lb cond prob} by  $\P(U_{m_{\Sc}}^{n_t}=u_{m_{\Sc}}^{n_t},U_{\a}^{n_t}=u_{\a}^{n_t})$ and taking the sum over all $u_{\a}^{n_t}\in\Uc_{\a}^{n_t}$ and $u_{m_{\Sc}}^{n_t}\in\Uc_{m_{\Sc}}^{n_t}$, we get
\begin{align}
{1\over n} H(U_{\a}^{n_t}|U_{m_{\Sc}}^{n_t})&\geq {1\over n}H(U_{\a}^{n_t}) + \sum_{u_{\a}^{n_t}\in\Uc_{\a}^{n_t}} \P(U_{\a}^{n_t}=u_{\a}^{n_t}) \log  {\P(U_{\a}^{n_t}=u_{\a}^{n_t})\over p_{U_{\a}}(u_{\a}^{n_t})} + {|\Sc|\log( 1-\e)\over n} -  2\e ({n_t\over n})(\sum\limits_{s\in\Sc}h_{k_s}^{(k)}),
\end{align}
or, rearranging the terms,
\begin{align}
{1\over n} I(U_{\a}^{n_t};U_{m_{\Sc}}^{n_t})+ \sum_{u_{\a}^{n_t}\in\Uc_{\a}^{n_t}} \P(U_{\a}^{n_t}=u_{\a}^{n_t}) \log  {\P(U_{\a}^{n_t}=u_{\a}^{n_t})\over p_{U_{\a}}(u_{\a}^{n_t})} \leq {|\Sc|\log( 1-\e)\over n} + 2\e ({n_t\over n})(\sum\limits_{s\in\Sc}h_{k_s}^{(k)}).\label{eq:I+D}
\end{align}
The second term on the left hand side of \eqref{eq:I+D} is the Kullback Leibler distance between $ \P(U_{\a}^{n_t}=u_{\a}^{n_t})$ and $p_{U_{\a}}(u_{\a}^{n_t})$ and hence is positive. Therefore,
\begin{align}
{1\over n} I(U_{\a}^{n_t};U_{m_{\Sc}}^{n_t}) \leq {|\Sc|\log( 1-\e)\over n} + 2\e({n_t\over n}) (\sum\limits_{s\in\Sc}h_{k_s}^{(k)}).\label{eq:bound-secrecy}
\end{align}
But letting $n$ grow to infinity and choosing $\e$ small enough, we can make the right hand side of \eqref{eq:bound-secrecy} arbitrarily small, for every $\a\in\Ac_{s}$.

\end{itemize}

\subsection{Proof of Theorem \ref{thm:outer_bd-variable}}
Let $\rv\in\Rc_0^{(w)}$ be  achievable. We prove that $\rv\in\Rc_{0,o}$.
Since $\rv\in\Rc_0^{(w)}$, there exists a family of weakly secure network codes of growing blocklength $n_{\ell}$, source rates $\rv^{(n_{\ell})}_{m_{\Sc}}$ and key rates $\rv^{(n_{\ell})}_{k_{\Sc}}$, and a sequence $(\e_{\ell})$, $\ell\geq 1$, which monotonically converges to zero, such that, for any source $s\in\Sc$,
\begin{align}
r_{m_s}^{(n_{\ell})}\geq r_s-\e_{\ell},\label{eq:source-rate-lb}
\end{align}
and for any set of edges $\a\in\Ac_s$,
\begin{align}\label{eq:sec-const}
{1\over n_{\ell}} I(M_{\Sc};W_{\a}) \leq \e_{\ell}.
\end{align}
Let $(M^{(\ell)}_s,K^{(\ell)}_s)_{s\in\Sc}$ and $(W^{(\ell)}_e)_{e\in\Ec}$ denote the source messages, keys, and edge messages of the code of blocklength $n_{\ell}$, respectively.

Corresponding to  each $\ell\in\mathds{N}$, define $N$ random variables $(Y^{(\ell)}_{m_s})_{s\in\Sc}, (Y^{(\ell)}_{k_s})_{s\in\Sc}$ and $(Y^{(\ell)}_{e})_{e\in\Ec}$ as follows.  For $s\in\Sc$,
\begin{align}\label{eq:Yms}
Y^{(\ell)}_{m_s}=M^{(\ell)}_s,
\end{align}
and
\begin{align}\label{eq:Yks}
Y^{(\ell)}_{k_s}=K^{(\ell)}_s,
\end{align}
and for $e\in\Ec$,
\begin{align}\label{eq:Ye}
Y^{(\ell)}_e=W^{(\ell)}_e.
\end{align}  Let $\hv^{(\ell)}$ denote the entropy function of the defined $N$ random variables. From the independence of the source messages and keys, it follows that
\[
H((M^{(\ell)}_s,K^{(\ell)}_s)_{s\in\Sc})=\sum_{s\in\Sc}(H(M^{(\ell)}_s)+H(K^{(\ell)}_s)).
\]
Hence, $\hv^{(\ell)}\in\Gamma_1$.  Consider message $W^{(\ell)}_e$ transmitted on edge $e=(v_1,v_2)\in\Ec$. If $v_1$ is a source node, \ie $v_1\in\Sc$, then $W^{(\ell)}_e$ is a function of source message $M^{(\ell)}_s$ and key $K^{(\ell)}_s$. In other words, $H(W^{(\ell)}_e|M^{(\ell)}_s,K^{(\ell)}_s)=0$, or equivalently $H(M^{(\ell)}_s,K^{(\ell)}_s,W^{(\ell)}_e)=H(M^{(\ell)}_s,K^{(\ell)}_s)$. Therefore, $\hv^{(\ell)}\in\Gamma_2$.   If $v_1$ is not a source message, then $W^{(\ell)}_e$ is a function of the messages traversing the incoming edges of node $v_1$,  and hence $\hv^{(\ell)}\in\Gamma_3$ as well.

The code is assumed to have zero probability of error. Therefore, for each sink node $t\in\Tc$, source messages in $\b(t)$ are a deterministic function of the incoming message of node $t$. That is, for $t\in\Tc$,
\[
H(M^{(\ell)}_{\b(t)}|W^{(\ell)}_{\In(t)})=H((M^{(\ell)}_s)_{s\in\b(t)}|(W^{(\ell)}_e)_{e\in\In(t)})=0.
\]
Hence, $\hv^{(\ell)}\in\Gamma_4$ as well. Clearly, $\hv^{(\ell)}$ is an entropic vector, \ie $\hv^{(\ell)}\in\Gamma_{N}^*$. Combining all these results, it follows that
\[
\hv^{(\ell)}\in \Gamma_{N}^{*}\cap \Gamma_{1234}.
\]
Moreover,  ${\bf 0}\in\Gamma_{N}^{*}\cap \Gamma_{1234}$. Therefore,  since  $\overline{\rm con}(\Gamma_{N}^{*}\cap \Gamma_{1234})$ is a convex set,
\[
(1-n_{\ell}^{-1}){\bf 0}+n_{\ell}^{-1}\hv^{(\ell)} = n_{\ell}^{-1}\hv^{(\ell)} \in   \overline{\rm con}(\Gamma_{N}^{*}\cap \Gamma_{1234}).
\]

As mentioned earlier, at blocklength $n_{\ell}$, the message transmitted on edge $e$, $W^{(\ell)}_e$, should satisfy $H(W^{(\ell)}_e)\leq n_{\ell}c_e$. Therefore, $n_{(\ell)}^{-1}\hv^{(\ell)}\in\Gamma_5$, and overall,
\begin{align}\label{eq:step1}
n^{-1}_{\ell}\hv^{(\ell)}\in\overline{\rm con}(\Gamma_{N}^{*}\cap \Gamma_{1234})\cap\Gamma_5.
\end{align}

Given the sequence, $\e_{\ell}$, $\ell=1,2,\ldots$, define a sequence of sets as
\begin{align}
\Gamma_{6}^{(\ell)} \triangleq \big\{\hv\in\Hc_{N}:\;  h_{m_{\Sc}}+h_{\a}-h_{m_{\Sc},\a} \leq \e_{\ell}, \; \forall\a\in\Ac_s \}.
\end{align}

Combining \eqref{eq:sec-const} and \eqref{eq:step1}, it follows that
\begin{align}
n_{\ell}^{-1}\hv^{(\ell)}\in\overline{\rm con}(\Gamma_{N}^{*}\cap \Gamma_{1234})\cap\Gamma_5\cap \Gamma_6^{(\ell)}.\label{eq:outer-bd-1}
\end{align}

To finish the proof, define
\begin{align}
\Rc_{o}^{(\ell)}\triangleq\Lambda( {\rm Proj}_{\Sc}(\overline{\rm con}(\Gamma_{N}^{*}\cap \Gamma_{1234})\cap\Gamma_5\cap \Gamma_6^{(\ell)})).\label{eq:outer-bd-2}
\end{align}
By \eqref{eq:source-rate-lb}, ${n_{\ell}^{-1}}h^{(\ell)}_s={n_{\ell}^{-1}}H(M_s^{(\ell)})\geq r_s-\e_{\ell}$ for all $s\in\Sc$. Combing this with \eqref{eq:outer-bd-1}, and \eqref{eq:outer-bd-2}, yields
\[
\rv-\e_{\ell}\in\Rc_{o}^{(\ell)}.
\]
Since $\e_{\ell}$ is monotonically decreasing to zero, and since construction of $\Rc_{o}^{(\ell)}$ involves applying the $\Lambda(\cdot)$ operation, we have
\[
\rv-\e_{\ell}\in\bigcap\limits_{\ell'=\ell}^{\infty}\Rc_{o}^{(\ell')}.
\]
On the other hand, for any $\ell\geq 1$,
\[
\bigcap\limits_{\ell'=\ell}^{\infty}\Rc_{o}^{(\ell')}=\Rc_o.
\]
Therefore,  $\rv-\e_{\ell}\in \Rc_o$, for $\ell\geq 1$. But, $\Rc_o$ is a closed set. Therefore, $\rv\in\Rc_o$ as well. This concludes the proof.


\section{Capacity region: asymptotically zero-error coding}\label{sec:c-region-e}

In this section we study  capacity regions of wiretap networks with asymptotically zero-error reconstructions  subject to weak secrecy. We derive inner and outer bounds on the capacity regions of such networks.

 \begin{theorem} [Inner bound]\label{thm:inner_bd-fixed}
Define
\[
\Rc_{\e,i}\triangleq \Lambda({\rm Proj}_{\Sc}(\overline{\rm con}(\Gamma_{N}^{*}\cap \Gamma_{1236})\cap\Gamma_{45})).
\]
Then,
\[
\Rc_{\e,i}\subseteq \Rc_{\e}^{(w)}.
\]
\end{theorem}

 \begin{theorem} [Outer bound]\label{thm:outer_bd-fixed}
Define
\[
\Rc_{\e,o}\triangleq \Lambda({\rm Proj}_{\Sc}(\overline{\rm con}(\Gamma_{N}^{*}\cap \Gamma_{123})\cap\Gamma_{456})).
\]
Then,
\[
\Rc_{\e}^{(w)}\subseteq \Rc_{\e,o}.
\]
\end{theorem}

Note that,  similar to $\Gamma_{12346}$,  $\Gamma_{1236}$ is also a linear subspace of $\Hc_N$. Therefore,  similar to Lemma \ref{lemma:1}, we have:
\begin{lemma}\label{lemma:2}
\[
\overline{\rm con}(\Gamma_{N}^{*}\cap \Gamma_{1236}) = \overline{D(\Gamma_{N}^{*}\cap \Gamma_{1236})}
\]
\end{lemma}

\subsection{Proof of Theorem \ref{thm:inner_bd-fixed}}
Similar to the proof of Theorem \ref{thm:inner_bd-fixed}, by Lemma \ref{lemma:2},  it suffices to show that  $\Rc'_{\e,i}\triangleq{\rm Proj}_{\Sc}(\overline{D(\Gamma_{N}^{*}\cap \Gamma_{1236})}\cap\Gamma_{45})\subseteq \Rc_{\e}^{(w)}$.

For any $\rv\in\Rc'_{\e,i}$, by definition of $\Rc'_{\e,i}$, there exists  sequences $\hv^{(k)}\in \Gamma_{N}^{*}\cap \Gamma_{123456}$,  $a_k\in(0,1]$, and $\e_k\in(0,1]$, $k=1,2,\ldots$, such that, for any $s\in\Sc$,
\begin{align}
a_k h^{(k)}_{m_s} \geq r_s-\e_k,\label{eq:source rate constraint-var}
\end{align}
for each edge $e\in\Ec$,
\begin{align}
a_kh^{(k)}_e\leq c_e+\e_k,
\end{align}
for each terminal $t\in\Tc$,
\begin{align}\label{eq:ce}
h^{(k)}_{m_{\b(t)},\In(t)}-h^{(k)}_{\In(t)}\leq \e_k,
\end{align}
and $\e_k\to 0$, as $k$ grows to infinity.

Similar to the proof of Theorem \ref{thm:inner_bd-variable},  we build a sequence of weakly secure network codes with asymptotically zero probability of error based on the sequence of entropic vectors  $\hv^{(k)}$.

Since $\hv^{(k)}\in\Gamma_{N}^{*}\cap \Gamma_{1236}$, there exists a set of finite-alphabet random variables random variables, $\{U_{m_s},U_{k_s}: s\in\Sc\}\cup\{U_e\; : \: e\in\Ec\}$, which satisfy conditions (i), (ii) and (iii) corresponding to equations \eqref{eq:i}, \eqref{eq:ii} and \eqref{eq:iii}, respectively, which are listed in the proof of Theorem \ref{thm:inner_bd-variable}. The only difference is that since $\hv^{(k)}$ is no longer in $\Gamma_4$, the  condition stated in (iv), which requires reconstruction of source variables of interest at each sink does not hold in this case.

\subsubsection{Random code construction}

Consider building the source and key codebooks according to the construction described in the proof of Theorem \ref{thm:inner_bd-variable}, \ie for $s\in\Sc$,  let $\Mc_s\triangleq \Tc_{\e}^{(n_t)}(U_{m_s})$, and $\Kc_s\triangleq \Tc_{\e}^{(n_t)}(U_{k_s})$, where $|\Mc_s|\triangleq2^{nr^{(k)}_{m_s}}$ and $|\Kc_s|\triangleq2^{nr^{(k)}_{k_s}}$. As before,  $n$ denotes the code's blocklength.

Source $s\in\Sc$ picks message $M_s$ and key $K_s$ uniformly at random from $\Mc_s$ and $\Kc_s$, respectively. Edge $e\in\Out(s)$, applies the function $\eta_e$, which is specified from  by the fact that $\hv^{(k)}\in\Gamma_2$,  symbol-by-symbol to $(U_{m_s}^{n_t},U_{k_s}^{n_t})$, the source builds  codeword $W_e=U_e^{n_t}$. In other words,  for $e\in\Out(s)$, and $i\in[1:n_t]$,
\begin{align}
U_{e,i}=\eta_e(U_{m_s,i},U_{k_s,i}).
\end{align}

After building the message and key codebooks in this manner, the encoding operations are performed as follows.  Each source $s\in\Sc$ chooses message $M_s$ and key $K_s$ uniformly at random from sets $\Mc_s$ and $\Kc_s$, respectively. Edge $e\in\Out(s)$ applies function $\eta_e$ symbol-by-symbol to $(M_s,K_s)=(U_{m_s}^{n_t},U_{k_s}^{n_t})$ and generates $U_e^{n_t}$, \ie $U_{e,i}=\eta_e(U_{m_s,i},U_{k_s,i})$, for $i\in[1:n_t]$.  Unlike the case of variable length coding, where always $W_e$ is equal to $U_e^{n_t}$, in this case, if $(U_{m_s}^{n_t},U_{k_s}^{n_t})\in\Tc_{\e}^{(n_t)}(U_{m_s},U_{k_s})$, then $W_e=U_e^{n_t}$; Otherwise, $W_e=0$. Note that by Lemma \ref{lemma:typical-g}, if $(U_{m_s}^{n_t},U_{k_s}^{n_t})\in\Tc_{\e}^{(n_t)}(U_{m_s},U_{k_s})$, then $U_e^{n_t}\in\Tc_{\e}^{(n_t)}(U_e)$ as well.

Similarly  node $v\in\Vc\backslash(\Sc\cup\Tc)$ constructs $W_e$, $e\in\Out(v)$, as a function of $(W_{e'}: e'\in\In(v))$ as follows. If there exists an incoming message in the received messages of node $v$, which is equal to $0$, then let $W_e=0$. Otherwise, node $v$ constructs $U_e^{n_t}$, such that $U_{e,i}=\eta_e(U_{e',i}:\; e'\in\In(v))$, $i\in[1:n_t]$. If $U_{\In(v)}^{n_t}\in\Tc_{\e}^{(n_t)}(U_{\In(v)})$, then $W_e=U_e^{n_t}$; Otherwise, $W_e=0$. Again, by Lemma \ref{lemma:typical-g}, if  $U_{\In(v)}^{n_t}\in\Tc_{\e}^{(n_t)}(U_{\In(v)})$ then  $U_e^{n_t}\in\Tc_{\e}^{(n_t)}(U_e)$ as well.

Finally, sink node $t\in\Tc$ reconstructs the message $M_s$, $s\in\b(t)$, as a function of $(W_e: e\in\In(t))$. Let  $\hat{M}_{s\to t}$ denote the reconstruction of message $M_s$ at $t$. If at least one of the incoming messages of node $v$ is $0$, then it lets $\hat{M}_{s\to t}$ equal to an arbitrary fixed message from $\Mc_s$.
Otherwise, it looks for $U_{m_s}^{n_t}\in\Mc_s$ such that
\[
(U_{m_s}^{n_t},(U_e^{n_t})_{e\in\In(t)})\in\Tc_{\e_t}^{(n_t)}(U_{m_s},U_{\In(t)}),
\]
and lets $\hat{M}_{s\to t}=U_{m_s}^{n_t}$.  If there is no such message, again it lets $\hat{M}_{s\to t}$ equal to an arbitrary predetermined message from $\Mc_s$.

Similar to \eqref{eq:n}, blocklength $n$ is defined as
\begin{align}
n \triangleq \left\lceil{n_t(1+\delta)\over a_k}\right\rceil,
\end{align}
where
\begin{align}
\delta \triangleq 2\max_{e\in\Ec}\left({(1+\e)(1+{\e_kc_e^{-1}})\over 1-\e c_e^{-1}}\right)-2.\label{eq:delta-2}
\end{align}
Again, for fixed $\e_k$, choosing $n_t$ large enough and $\e$ small enough, $\d$ can be made arbitrarily small.

\subsubsection{Performance evaluation}

\begin{enumerate}

\item[i.] Source rates:
Since $\Mc_s= \Tc_{\e}^{(n_t)}(U_{m_s})$,  for $n_t$  large enough, from \eqref{eq:size-typical-set} it follows that
 \begin{align}
r^{(k)}_{m_s}&={\log| \Tc_{\e}^{(n_t)}(U_{m_s})|\over n}\nonumber\\
&  \geq {n_t(1-\e)\over n_t(1+\delta)+a_k}(r_s-\e_k)+ {a_k\log (1-\e)\over n_t(1+\delta)+a_k}.\label{eq:lower-bd-rs-2}
 \end{align}
Again tuning the parameters $\e_k$, $\e$ and $n_t$ appropriately, we can make the right hand side of \eqref{eq:lower-bd-rs-2} arbitrarily close to $r_s$.

\item[ii.] Channel capacity constraints: In this case, we have assumed fixed length coding. Hence, to prove that channel capacity  constraints are satisfied, it suffices to show that $n^{-1}\log|\Wc_e|\leq nc_e$. But, $\Wc_e=\Tc_{\e}^{(n_t)}(U_e)\cup\{0\}$. Therefore,
\begin{align}
n^{-1}\log|\Wc_e|&=n^{-1}\log(|\Tc_{\e}^{(n_t)}(U_e)|+1)\nonumber\\
&\stackrel{(a)}{\leq} n^{-1}\log|\Tc_{\e}^{(n_t)}(U_e)| +n^{-1}|\Tc_{\e}^{(n_t)}(U_e)|^{-1}\nonumber\\
&\leq {a_k(1+\e)\over 1+\d} h_e^{(k)}+\e\nonumber\\
&\leq ({1+\e\over 1+\d} )(c_e+\e_k)+\e\nonumber\\
&\stackrel{(b)}{\leq} c_e,
\end{align}
where  (a) follows from $\log(1+x)\leq x$, and (b) follows from \eqref{eq:delta-2}.

\item[iii.] Security guarantee: We  prove  that, for any $\a\in\Ac_s$, $n^{-1} I(M_{\Sc};W_{\a})$ can be made arbitrarily small, as blocklength grows to infinity. Note that the  difference between the coding performed on the edges in this case and the case studied in the proof of Theorem \ref{thm:inner_bd-variable} is that, here,  codeword $W_e$ sent over edge $e$ can sometimes be  empty . This happens when either one of the incoming messages to the tail of $e$ is empty, or $U_{\In({\rm tail(e)})}^{n_t}$   is not typical.  In fact the set of messages carried by the edges in this case can be viewed as  the previous set of messages corrupted by an erasure channel with memory; If an input message of a node is empty, then all the messages carried by its descendants  are zero as well.

Let
\[
I_e \triangleq \ind_{W_e\neq 0},
\]
\ie $I_{e}$ is an indicator function that shows whether  the messages carried by edge $e$ is empty or not.  Note that $W_{e}$ is a function of $(U_{e}^{n_t},I_{e})$: $W_{e}=U_{e}^{n_t}$, if $I_e=1$, and  $W_e=0$, otherwise. Therefore,
\begin{align}
n^{-1} I(M_{\Sc};W_{\a})&\leq n^{-1} I(M_{\Sc};U^{(n_t)}_{\a},I_{\a})\nonumber\\
&\leq n^{-1} I(M_{\Sc};U_{\a}^{n_t})+n^{-1} I(M_{\Sc};I_{\a}|U_{\a}^{n_t})\nonumber\\
&\leq n^{-1} I(M_{\Sc};U_{\a}^{n_t})+n^{-1}\log |\a|.
\end{align}
In the proof of Theorem \ref{thm:inner_bd-variable}, we showed that  $n^{-1} I(M_{\Sc};U_{\a}^{n_t})$ can be made arbitrarily small. Since $\max_{\a\in\Ac_s} |\a|$  is fixed and does not grow with $n$, it follows that we can make $n^{-1} I(M_{\Sc};W_{\a})$ arbitrarily small for any $\a\in\Ac_s$.

\item[iv.] Probability of error: Proving that the  probability of error converges  to zero as $n$ grows to infinity is involved, but follows from the straightforward extension of the analysis presented  in \cite{YanY:12}, and we skip that here.

\end{enumerate}

\subsection{Proof of Theorem \ref{thm:outer_bd-fixed}}

Let $\rv\in\Rc_{\e}^{(w)}$. We show that $\rv\in\Rc_{\e,o}$. Since $\rv\in\Rc_{\e}^{(w)}$, there exists a family of weakly secure network codes of increasing blocklength $n_{\ell}$ with message rates $\rv^{(n_{\ell})}_{m_{\Sc}}$ and key rates $\rv^{(n_{\ell})}_{k_{\Sc}}$, and a sequence $(\e_{\ell}: \ell=1,2,\ldots)$,  monotonically converging to zero, such that
\begin{align}
r_{m_s}^{(n_{\ell})}\geq r_s-\e_{\ell},\label{eq:source-rate-lb-eps}
\end{align}
for all  $s\in\Sc$,
\begin{align}\label{eq:sec-const-eps}
{1\over n_{\ell}} I(M_{\Sc};W_{\a}) \leq \e_{\ell},
\end{align}
for all $\a\in\Ac_s$, and
\begin{align}\label{eq:p-error-eps}
\P(\hat{M}_{s\to t}&\neq M_{s}, {\rm for}\;{\rm some}\;t\in\Tc,\;s\in\b(t)) \leq \e_{\ell}.
\end{align}

For blocklength $n_{\ell}$, let $(M^{(\ell)}_s,K^{(\ell)}_s)_{s\in\Sc}$ and $(W^{(\ell)}_e)_{e\in\Ec}$ denote its set of source messages, keys, and edge messages, respectively. Also, corresponding to  each $\ell$, define $N$ random variables $(Y^{(\ell)}_{m_s})_{s\in\Sc}, (Y^{(\ell)}_{k_s})_{s\in\Sc}$ and $(Y^{(\ell)}_{e})_{e\in\Ec}$ as \eqref{eq:Yms}, \eqref{eq:Yks}, and \eqref{eq:Ye}, and similarly let $\hv^{(\ell)}$ denote the entropy function of these $N$ random variables. By the same arguments presented in the proof of Theorem \ref{thm:outer_bd-variable},
\[
 n_{\ell}^{-1}\hv^{(\ell)} \in   \overline{\rm con}(\Gamma_{N}^{*}\cap \Gamma_{123}).
\]

In the fixed-length channel coding paradigm, the Message $W^{(\ell)}_e$ transmitted on edge $e$ satisfies $\log |\Wc^{(\ell)}_e|\leq n_{\ell}c_e$. Therefore, $H(\Wc^{(\ell)}_e)=\log |\Wc^{(\ell)}_e|\leq n_{\ell}c_e$. Hence, $n_{\ell}^{-1}\hv^{(\ell)}\in\Gamma_5$, and
\begin{align}\label{eq:step1-eps}
n^{-1}_{\ell}\hv^{(\ell)}\in\overline{\rm con}(\Gamma_{N}^{*}\cap \Gamma_{1234})\cap\Gamma_5.
\end{align}

For every sink node $t\in\Tc$, combining Fano's inequality and \eqref{eq:p-error-eps} yields
\begin{align}
H(M^{(\ell)}_{\b(t)}|W^{(\ell)}_{\In(t)})&\leq 1+n_{\ell}(\sum_{s\in\b(t)}r_{m_s}^{(n_{\ell})})\P(\hat{M}_{s\to t}\neq M_{s}, {\rm for}\;{\rm some}\;s\in\b(t))\nonumber\\
&\leq 1+n_{\ell}(\sum_{s\in\b(t)}r_{m_s}^{(n_{\ell})})\e_{\ell},\nonumber
\end{align}
or dividing both sides by $n_{\ell}$
\begin{align}
n_{\ell}^{-1}H(M^{(\ell)}_{\b(t)}|W^{(\ell)}_{\In(t)}) &\leq n_{\ell}^{-1}+ \sum_{s\in\b(t)}r_{m_s}^{(n_{\ell})}\e_{\ell}.\label{eq:bound-fano-conditional}
\end{align}

For any source $s\in\Sc$, if $r_{m_s}^{(n_{\ell})}\geq \sum_{e\i\Out(s)}c_e$, then the probability of error gets arbitrarily close to one, for $n$ large enough. Hence, without loss of generality, we assume that, the rate of source $s\in\Sc$ does not exceed its outgoing  edges  sum capacity. Then,
\begin{align}
\sum_{s\in\b(t)}r_{m_s}^{(n_{\ell})}&\leq \sum_{s\in\Sc}r_{m_s}^{(n_{\ell})}\nonumber\\
&\leq \sum_{s\in\Sc}\sum_{e\in\Out(s)}c_e\nonumber\\
&\triangleq c_M.\label{eq:cM}
\end{align}
Combining \eqref{eq:bound-fano-conditional} and \eqref{eq:cM}, it follows that
\begin{align}
n_{\ell}^{-1}H(M^{(\ell)}_{\b(t)}|W^{(\ell)}_{\In(t)}) &=h_{m_{\b(t)},\In(t)}- h_{\In(t)}\nonumber\\
&\leq n_{\ell}^{-1}+c_M\e_{\ell},
\end{align}
for every $t\in\Tc$.

Given the sequence, $(\e_{\ell}: \ell=1,2,\ldots)$, define the two sequences of sets as
\begin{align}
\Gamma_{4}^{(\ell)} \triangleq \big\{\hv\in\Hc_{N}:\;  h_{m_{\b(t)},\In(t)}- h_{\In(t)}\leq n_{\ell}^{-1}+c_M\e_{\ell}, \;\forall\;t\in\Tc\},\label{eq:gamma4}
\end{align}
and
\begin{align}
\Gamma_{6}^{(\ell)} \triangleq \big\{\hv\in\Hc_{N}:\;  h_{m_{\Sc}}+h_{\a}-h_{m_{\Sc},\a} \leq \e_{\ell}, \; \forall\a\in\Ac_s \}.\label{eq:gamma6}
\end{align}
Note that $n_{\ell}^{-1}\hv^{(\ell)}\in\Gamma_{4}^{(\ell)}\cap\Gamma_{6}^{(\ell)}$. Therefore, from  \eqref{eq:step1-eps}, it follows that
\begin{align}
n_{\ell}^{-1}\hv^{(\ell)}\in\overline{\rm con}(\Gamma_{N}^{*}\cap \Gamma_{123})\cap\Gamma_4^{(\ell)}\cap \Gamma_5\cap \Gamma_6^{(\ell)}.\label{eq:outer-bd-1}
\end{align}

The rest of the proof can be done similar to the last steps of the proof of Theorem \ref{thm:outer_bd-variable}. Let
\begin{align}
\Rc_{\e,o}^{(\ell)}\triangleq\Lambda( {\rm Proj}_{\Sc}(\overline{\rm con}(\Gamma_{N}^{*}\cap \Gamma_{123})\cap\Gamma_4^{(\ell)}\cap \Gamma_5\cap \Gamma_6^{(\ell)})).\label{eq:outer-bd-2-eps}
\end{align}
By \eqref{eq:source-rate-lb-eps}, ${n_{\ell}^{-1}}h^{(\ell)}_s\geq r_s-\e_{\ell}$ for all $s\in\Sc$. Combining this with \eqref{eq:outer-bd-1} and \eqref{eq:outer-bd-2-eps}, yields
\[
\rv-\e_{\ell}\in\Rc_{o}^{(\ell)}.
\]
By the same argument we had before,
\[
\rv-\e_{\ell}\in\bigcap\limits_{\ell'=\ell}^{\infty}\Rc_{\e,o}^{(\ell')}.
\]
On the other hand, for any $\ell\geq 1$,
\[
\bigcap\limits_{\ell'=\ell}^{\infty}\Rc_{\e,o}^{(\ell')}=\Rc_{\e,o}.
\]
Therefore,  $\rv-\e_{\ell}\in \Rc_o$, for $\ell\geq 1$. Since $\Rc_o$ is a closed set, $\rv\in\Rc_o$ as well.

\section{Conclusion}\label{sec:conclusion}

In this paper, we considered the problem of multi-source multi-sink communication over  wiretap networks.  We studied both zero and asymptotically zero probability of  reconstruction error, and  proved that in the case of asymptotically zero error,  the capacity region  subject to weak  secrecy requirement is equal to the capacity region subject to strong secrecy requirement. In fact we  proved this equivalence  for general multi-source multi-destination wireless networks modeled by discrete memoryless channels. This result implys that equivalence models derived previously for weak secrecy \cite{DikaliotisY:12,Ted_thesis} also hold under strong secrecy constraint.  Both for zero-error and asymptotically zero-error communication, we  derived inner and outer bounds on the capacity region subject to weak secrecy constraint   in terms of the intersection of the entropic region and some hyperplanes defined by the network's constraints.

\section*{Acknowledgment}
This work is partially supported by NSF grant CNS-0905615.

\renewcommand{\theequation}{A-\arabic{equation}}
\setcounter{equation}{0}  

\section*{APPENDIX A: Proof of Theorem \ref{thm:strong_eq_weak}}  \label{app1} 

We start by a weakly secure code and apply it multiple times. Then  we use extractor functions to build a strongly secure  code. The difference with \cite{MaurerW:00} is that here we have multiple sources.

Clearly, $\Rc^{(w)}(p(y^m|x^m),\Sc,\Ac)\subseteq\Rc^{(s)}(p(y^m|x^m),\Sc,\Ac)$. We only need to prove that
\[
\Rc^{(s)}(p(y^m|x^m),\Sc,\Ac)\subseteq\Rc^{(w)}(p(y^m|x^m),\Sc,\Ac).
\]
To achieve this goal, we start by a weakly secure code, and apply the same code $L$ times. Then, employing extractor functions to the concatenation of the messages, we construct a strongly secure code.

Let $\rv=(r_s: s\in\Sc)\in\Rc^{(w)}(p(y^m|x^m),\Sc,\Ac)$. By definition, for any $\e>0$ and $n$ large enough, there exists a weakly secure code of blocklength $n$ and source coding rates $(r_s^{(n)}: s\in\Sc)$, such that $r^{(n)}_s\geq r_s-\e$, for every $s\in\Sc$,  $\P(M_s\neq \Mh_{s\to t})\leq \e$, for $s\in\Sc$ and $t\in\Dc_s$, and $n^{-1}I(Y_{\a}^n;M_{\Sc})\leq \e$, for all $\a\in\Ac$.

Consider applying this code $L$ times. For $s\in\Sc$, let $M_{s,\ell}$ and $\Mh_{s\to t,\ell}$ denote the message transmitted by source $s$ and its reconstruction at node $t\in\Dc_s$, respectively, at session $\ell\in[1:L]$.
By the union bound, $\P(\hat{M}_{s\to t}^L\neq M_s^L)\leq \e L$.  Consider modeling  the problem as a problem of source coding with side information. Source $s$ has access to i.i.d.~ samples of $M_s$ and terminal $t$ observes correlated side information $\Mh_{s\to t}$. By the Fano's inequality  \cite{cover}, $H(M_s|\Mh_{s\to t})\leq nr_{s}^{(n)}\e+1$. Hence, by the Slepian-Wolf coding theorem \cite{SlepianW:73}, if $L$ is large enough, for any $\e_1>0$, source $s$ by sending $LH(M_s|\Mh_{s\to t}) \leq L(nr_{s}^{(n)}\e+1)$ extra bits to all terminal  in $\Tc_s\triangleq\{t: s\in\Dc_s\}$, can ensure that
\[
\P(\Mt_{s\to t}^L\neq M_s^L)\leq \e_1,
\]
where $\Mt_{s\to t}^L$ denotes the reconstruction of $M_s^L$ at terminal $t$ as a function of $\Mh_{s\to t}^L$ and the extra information received from the Slepian-Wolf coding sessions.  Let $O_s\in{\cal O}_s$ denote the message sent from source $s$ to all terminals $t$ such that $s\in\b(t)$. Note that $\log |{\cal O}_s|\leq L(nr_{s}^{(n)}\e+1)$.

For  $\e_2>0$, define event $\Bc$ as
\[
\Bc\triangleq \{(M^L_{\Sc},Y^{nL}_{\a})\in\Tc_{\e_2}^{(L)}(M_{\Sc},Y^n_{\a}): {\rm for}\;{\rm all}\;\a\in\Ac\}.
\]
By the Hoeffding's inequality \cite{Hoeffding}, for $(m_{\Sc},y^{n}_{\a})\in\prod_{s\in\Sc}\Mc_{s}\times \prod_{i\in\a}\Yc_i^n$,
\begin{align}
\P&\left(|\pi(m_{\Sc},y^{n}_{\a}|M^L_{\Sc},Y^{nL}_{\a})-p(m_{\Sc},y^n_{\a})|\leq \e_2 p(m_{\Sc},y_{\a}^n)\right) \nonumber\\
&\geq 1- 2^{-2L\e_2^2p(m_{\Sc},y^n_{\a})+1}.\label{eq:hoeffding}
\end{align}
Combining \eqref{eq:hoeffding} with the union bound yields
\begin{align}
\P(\Bc^c)&\leq \sum_{\a\in\Ac}\sum_{\substack{(m_{\Sc},y^n_{\a})\in \\ \prod\limits_{s\in\Sc}\Mc_{s}\times \prod\limits_{i\in\a}\Yc_i^n}}2^{-2L\e_2^2p(m_{\Sc},y^n_{\a})+1}\nonumber\\
&\leq |\Ac|\Big(2^{n\sum\limits_{s\in\Sc}r_{s}^{(n)} }\prod_{i\in\Vc}|\Yc_i|^n\Big) 2^{-2L\e_2^2p^*+1},\nonumber
\end{align}
where $p^*\triangleq\min_{\a\in\Ac}\min_{(m_{\Sc},y^n_{\a})} p(m_{\Sc},w_{\a})$. Note that  $p^*>0$ and does not depend on $L$. Hence,
\begin{align}
\P(\Bc^c)\leq 2^{-\g L+\eta},
\end{align}
where $\g\triangleq 2\e_2^2p^*$, $p^*\triangleq\min_{\a\in\Ac}\min_{(m_{\Sc},y^n_{\a})} p(m_{\Sc},y^n_{\a})>0$, and $\eta\to 0$ as $L\to \infty$.

Since the original code is assumed to be weakly secure, for any $\a\in\Ac$, $I(M_{\Sc};Y^n_{\a})<\e n$. Therefore, for any $s\in\Sc$, $H(M_s|Y^n_{\a},M_{\Sc\backslash s})\geq nr^{(n)}_{s}-n\e$, which due to the independence of the messages yields
\[
H(M_s|Y^n_{\a},M_{\Sc\backslash s})\geq nr_{s}^{(n)}-n\e.
\]
Also, for any $\a\in\Ac$, $y^{nL}_{\a}\in\prod_{i\in\a}\Yc_i^{nL}$, if $(m_{\Sc}^L,y_{\a}^{nL})\in\Tc_{\e_2}^{(L)}(M_{\Sc},Y_{\a}^n)$, then
\[
p(m_s^L|m_{\Sc\backslash s}^L,y_{\a}^{nL})\leq 2^{-L(1-\e_2)H(M_s|M_{\Sc\backslash s},Y_{\a}^{n})}.
\]
Therefore,
\begin{align}
&H_{\infty}(M_s^L|M^L_{\Sc\backslash s}=m^L_{\Sc\backslash s},Y_{\a}^{nL}=y^{nL}_{\a},\Bc)\nonumber\\
&\geq (1-\e_2)LH(M_s|M_{\Sc\backslash s},Y^n_{\a})\nonumber\\
& \geq (1-\e_2)Ln(r_{s}^{(n)}-\e).
\end{align}

The error correcting messages $O_{\Sc}=(O_s: s\in\Sc)$ are sent to the desired users  without  additional  coding to prevent information leakage to  the adversary.   The adversary  observes random variables correlated by  $O_{\Sc}$  through outputs $Y_{\a}$.  By the data processing inequality, the worst case performance from the vewipoint of security  is when the adversary observes  the we messages $O_{\Sc}=(O_s: s\in\Sc)$ instead of their noisy version. For $(m^L_{\Sc\backslash s},y_{\a}^{nL})\in\Tc_{\e_2}^{(L)}(M_{\Sc},Y^n_{\a})$, by Lemma \ref{lemma:H-infty}, for any $\l>0$, with probability exceeding $1-2^{-\l}$,
\begin{align}
&H_{\infty}(M_s^L|M^L_{\Sc\backslash s}=m^L_{\Sc\backslash s},Y_{\a}^{nL}=y_{\a}^{nL},O_{\Sc}=o_{\Sc},\Bc)\nonumber\\
&\geq (1-\e_2)Ln(r_{s}^{(n)}-\e)-\sum_{s\in\Sc}\log|\Oc_s|-\l \nonumber\\
&\geq (1-\e_2)Ln(r_{s}^{(n)}-\e) -L(n\sum_{s'\in\Sc}r_{s'}^{(n)}\e+|\Sc|)-\l\nonumber\\
&= Lnr^{(n)}_{s}(1-\e_3),\label{eq:H-inf-lb}
\end{align}
where $\e_3\triangleq1-(1-\e_2)(1-\e/r_s^{(n)})-(\e\sum_{s'}r^{(n)}_{s'}-n^{-1})/r^{(n)}_{s}-\l/(Lnr_{s}^{(n)})$, and can be made arbitrarily small by choosing $L$ and $n$ large enough and $\e$ and $\e_1$ appropriately small. Let $\l\triangleq n+ \lceil \log L \rceil$, and define event $\Hc_s$ as the event  \eqref{eq:H-inf-lb} holds. Note that $\P(\Hc_s)\geq 1-2^{-n}L^{-1}$.

As mentioned before, to convert the weakly secure code into a strongly secure one, we apply Lemma \ref{lemma:8-MU}. Let $V_s$ be an independent random variable generated at source $s\in\Sc$, such that $\Vc_s=\{0,1\}^{n_{2,s}}$. By Lemma \ref{lemma:8-MU} and \eqref{eq:H-inf-lb}, for each source $s\in\Sc$ and any given $\d_1,\d_2>0$, there exists an extractor function $E_s:\{0,1\}^{Ln r^{(n)}_{s}}\times \{0,1\}^{n_{s,2}}\to \{0,1\}^{n_{s,3}}$, such that $n_{s,2}\leq \d_1Ln r^{(n)}_{s}$, $n_{s,3}\geq (1-\e_3-\d_2)Ln r^{(n)}_{s}$,
\begin{align}
&H(E_s(M_s^L,V_s)|Y_{\a}^{nL}=y_{\a}^{nL},M^L_{\Sc\backslash s}=m^L_{\Sc\backslash s},\Oc_{\Sc}=o_{\Sc},V_s,\Bc,\Hc_s)\nonumber\\
&\geq n_{s,3}-2^{-\sqrt{Lnr^{(n)}_{s}}-\d_s},\label{eq:lower-db-Es}
\end{align}
where $\d_s=o(1)$. Let $n_{s,3}\triangleq (1-\d_4)Ln r^{(n)}_{s}$.

Let $\bar{M}_s\triangleq E_s(M_s^L,V_s)$. From the definition of $E_s$, $\bar{M}_s\in\bar{\Mc}_s\triangleq [1:2^{n_{s,3}}]$, and
\begin{align}
&H(\bar{M}_{\Sc}|Y^{nL}_{\a},O_{\Sc},V_{\Sc}){\geq}  \sum_{s\in\Sc}H(\bar{M}_s|\Mt_{\Sc\backslash s},Y^{nL}_{\a},O_{\Sc},V_{\Sc})\nonumber\\
&\stackrel{(a)}{\geq}  \sum_{s\in\Sc}H(\bar{M}_s|M_{\Sc\backslash s}^L,Y^{nL}_{\a},O_{\Sc},V_{\Sc})\nonumber\\
&\stackrel{(b)}=  \sum_{s\in\Sc}H(\bar{M}_s|M_{\Sc\backslash s}^L,Y^{nL}_{\a},O_{\Sc},V_s)\nonumber\\
&{\geq}  \sum_{s\in\Sc}H(\bar{M}_s|M_{\Sc\backslash s}^L,Y^{nL}_{\a},O_{\Sc},V_s,\ind_{\Bc\cap\Hc_s})\nonumber\\
&\geq\sum_{s\in\Sc}H(\bar{M}_s|M^L_{\Sc\backslash s},Y^{nL}_{\a},O_{\Sc},V_s,\Bc,\Hc_s)\P(\Bc\cap\Hc_s)\label{eq:M-bad-s-lb}
\end{align}
where (a)  holds because $\bar{M}_s$ is a deterministic function of $(M_s^L,V_s)$, and (b) holds since $V_{\Sc\backslash s}$ is independent of  $(\bar{M}_s,M_{\Sc\backslash s}^L,Y^{nL}_{\a}$, $O_{\Sc},V_s)$. Since $\P(\Bc\cap\Hc_s)\geq 1- \P(\Bc^c)- \P(\Hc_s^c)$, it follows from \eqref{eq:M-bad-s-lb} that
\begin{align}
&H(\bar{M}_{\Sc}|Y^{nL}_{\a},O_{\Sc},V_{\Sc})\nonumber\\
&\geq\sum_{s\in\Sc}H(\bar{M}_s|M^L_{\Sc\backslash s},Y^{nL}_{\a},O_{\Sc},V_s,\Bc,\Hc_s)(1-2^{-n}L^{-1}-2^{-\g L+\eta})\nonumber\\
						   &\stackrel{(a)}{\geq}\sum_{s\in\Sc}(n_{s,3}-2^{-\sqrt{Lnr_s^{(n)}}-\d_s})(1-2^{-n}L^{-1}-2^{-\g L+\eta})\nonumber\\
						   &\geq \sum_{s\in\Sc}n_{s,3}-\e_4,\label{eq:22}
\end{align}
where (a) follows from \eqref{eq:lower-db-Es}, and $\e_4$ can be made arbitrarily small. Following the same steps used in the above equations for each $s\in\Sc$ yields
$H(\bar{M}_s|Y^{nL}_{\a},O_{\Sc},V_{\Sc})\geq n_{s,3}-\e_4$, and therefore, $H(\bar{M}_s)\geq H(\bar{M}_s|Y^{nL}_{\a},O_{\Sc},V_{\Sc})\geq n_{s,3}-\e_4$. Let $p_{u_s}$  denote the pdf corresponding to the uniform distribution over $\bar{\Mc}_s$. Then, $D(p_{\bar{m}_s}\|p_{u_s})= n_{s,3}-H(M_s)\leq \e_4$, and by Pinsker's inequality \cite{cover}, $d_{\rm TV}(p_{\bar{m}_s},p_{u_s})\leq \sqrt{0.5\e_4}$. Therefore, by the maximal coincidence theorem (Theorem 1.1. in Chapter 4 of \cite{Bremaud:99-book}), there exists a joint distribution $p_{u_{s},\bar{m}_s}$, such that if $(U_s,\bar{M}_s)\sim p_{u_{s},\bar{m}_s}$, then $\P(U_s\neq \bar{M_s}) = d_{\rm TV}(p_{\bar{m}_s},p_{u_s}) \leq \sqrt{0.5\e_4}$. Given $\bar{M}_s$, let $U_s$, the message at source $s$,  be the random variable generated from the conditional distribution $p_{u_s|\bar{m}_s}$. Since, for any $\a\in\Ac_s$, $U_{\Sc}\to\bar{M}_{\Sc}\to W_{\a}$, therefore, by the data processing inequality \cite{cover},
\[
I(U_{\Sc};Y^{nL}_{\a}) \leq I(\bar{M}_{\Sc};Y^{nL}_{\a})\leq \e_4,
\]
where the last step follows from \eqref{eq:22}.

Finally, we need to compute the effect of the extra communication required for carrying messages $(O_s: s\in\Sc)$ and $(V_s: s\in\Sc)$. To achieve this goal, we model the correlation between $M_s$ and $\{\Mh_{s\to t}\}_{t\in\Dc_s}$ as a broadcast channel. Since $\P(M_s\neq \Mh_{s\to t}, {\rm for}\;{\rm some}\; t\in\Dc_s)\leq |\Dc_s|\e$, and
\[
\P(M_s\neq \Mh_{s\to t}, {\rm for}\;{\rm some}\; t\in\Dc_s)=\sum_{m_{\Sc\backslash s}\in\prod_{v\in\Sc\backslash s}\Mc_{v}}\P(M_s\neq \Mh_{s\to t}, {\rm for}\;{\rm some}\; t\in\Dc_s|M_{\Sc\backslash s}=m_{\Sc\backslash s})p(m_{\Sc\backslash s}),
\]
there exists $m^*_{\Sc\backslash s}\in\prod_{v\in\Sc\backslash s}\Mc_{v}$, such that $\P(M_s\neq \Mh_{s\to t}, {\rm for}\;{\rm some}\; t\in\Dc_s|M_{\Sc\backslash s}=m^*_{\Sc\backslash s})\leq |\Dc_s|\e$. Again by the Fano's inequality, since the messages are independent,
\begin{align}
I(M_s;\Mh_{s\to t}|M_{\Sc\backslash s}=m^*_{\Sc\backslash s})&\geq nr_s^{(n)}(1-|\Dc_s|\e)-1\nonumber\\
&\geq n(r_s-\e)(1-|\Dc_s|\e)-1,\label{eq:mutual-info-Ms}
\end{align}
for all $t\in\Dc_s$.  Fixing the input messages of nodes in $\Sc\backslash s$ to $m^*_{\Sc\backslash s}$, consider the broadcast channel from node $s$ to  nodes in  $\Dc_s$, described by $p((\mh_{s\to t})_{ t\in\Dc_s}|m_{s}, M_{\Sc\backslash s}= m_{\Sc\backslash s})$.  Since \eqref{eq:mutual-info-Ms} holds for every $t\in\Dc_s$, the common rate capacity of this cannel is at least $n(r_s-\e)(1-|\Dc|\e)-1$. To achieve this rate, node $s$ generates its messages from uniform distribution over $\Mc_s$. Repeating this process $|\Sc|$ times,  all the required extra messages, \ie $O_{\Sc}$ and $V_{\Sc}$, are carried to their intended receivers with arbitrarily small probability of error. Overall, these operations increase the blocklength from $Ln$ by
\begin{align*}
\sum_{s\in\Sc}{\d_1Lnr_{s}^{(n)}+ L(nr_{s}^{(n)}\e+1)\over (r_s-\e)(1-|\Dc|\e)-n^{-1}}\leq Ln\d_2,
\end{align*}
where $\d_2$ can be made arbitrarily small.

 Let $\bar{M}_{s\to t}$, $V_{s\to t}$ and $O_{s\to t}$ denote the reconstructions of $U_s$, $V_s$ and $O_s$, at sink node $t\in\Dc_s$, respectively. By the union bound,
\begin{align}
\P(U_s\neq \bar{M}_{s\to t}) &\leq \P(U_s\neq \bar{M}_s)+\P(M^L_s \neq \Mt^L_{s\to t})+\P((V_s,O_s) \neq (V_{s\to t},O_{s\to t})),
\end{align}
which can be made arbitrarily small.

\bibliographystyle{unsrt}
\bibliography{myrefs}
\end{document}